\documentclass[final,11pt]{article}
\pdfoutput=1

\def\full{1}

\title{\textsc{\textbf{
      Title
    }}
}

\author{
  Author
}
\date{\today}

\usepackage{color,graphicx}

\usepackage{amsmath,amstext}

\usepackage{amsfonts,amssymb}

\usepackage{amsthm}

\usepackage{textcomp}

\newtheorem{theorem}{Theorem}[section]

\newtheorem{proposition}[theorem]{Proposition}
\newtheorem{lemma}[theorem]{Lemma}
\newtheorem{corollary}[theorem]{Corollary}
\newtheorem{conjecture}[theorem]{Conjecture}

\newtheorem{fact}[theorem]{Fact}

\theoremstyle{definition}
\newtheorem{definition}[theorem]{Definition}

\newtheorem{remark}[theorem]{Remark}

\usepackage{microtype}

\usepackage{url}

\usepackage{fullpage}

\usepackage{bm}

\usepackage[colorlinks,pagebackref,linkcolor=blue, filecolor = blue,
citecolor = blue, urlcolor = blue]{hyperref}

\usepackage{xspace}

\usepackage{dsfont}

\usepackage{nicefrac}

\newcommand{\nfrac}{\nicefrac}
\newcommand{\half}{\nfrac12}

\usepackage[color]{showkeys}

\newcommand{\mper}{\,.}
\newcommand{\mcom}{\,,}

\definecolor{DSgray}{cmyk}{0,0,0,0.7}

\definecolor{DSred}{cmyk}{0,0.7,0,0.7}

\newcommand{\Brac}[1]{\left[#1 \right]}

\newcommand{\poly}{\mathrm{poly}}
\newcommand{\vol}{\mathrm{vol}}
\newcommand{\dist}{\mathit{dist}}

\renewcommand{\leq}{\leqslant}
\renewcommand{\le}{\leqslant}
\renewcommand{\geq}{\geqslant}
\renewcommand{\ge}{\geqslant}

\newcommand{\Z}{\mathbb Z}

\newcommand{\R}{\mathbb R}

\newcommand{\Esymb}{\mathbb{E}}
\newcommand{\Psymb}{\mathbb{P}}

\DeclareMathOperator*{\E}{\Esymb}

\DeclareMathOperator*{\ProbOp}{\Psymb r}

\renewcommand{\Pr}{\ProbOp}

\newcommand{\Ex}[1]{\E\Brac{#1}}

\renewcommand{\epsilon}{\varepsilon}

\newcommand{\e}{\epsilon}
\newcommand{\eps}{\epsilon}

\newcommand{\sse}{\subseteq}

\usepackage{boxedminipage}

\ifnum\full=1
\newcommand{\onlyproc}[1]{}

\else
\newcommand{\onlyproc}[1]{#1}

\fi

\begin{document}

\sloppy


\newenvironment{enumerate*}%
  {\begin{enumerate}%
    \setlength{\itemsep}{0pt}%
    \setlength{\parskip}{0pt}}%
  {\end{enumerate}}

\newcommand{\dee}{\mathcal{D}}
\newcommand{\Range}{\mathcal{R}}
\newcommand\rd{\mathrm{d}}
\newcommand\Rn{\ensuremath{\Re^n}}
\newcommand\Rk{\ensuremath{\Re^k}}
\newcommand\Rd{\ensuremath{\Re^d}}
\renewcommand\R{\Re}
\newcommand\GammaDist{\mathrm{Gamma}}
\newcommand\err{\mathrm{err}}
\renewcommand\vol{\mathrm{Vol}}
\newcommand{\vollb}{\mathrm{VolLB}}
\newcommand{\Gvollb}{\mathrm{GVolLB}}

\title{\bf On the Geometry of Differential Privacy}
\author{Moritz Hardt\thanks{Department of Computer Science, Princeton
University, \texttt{mhardt@cs.princeton.edu}. 
Supported by NSF grants CCF-0426582 and CCF-0832797.
Part of this work has been done while the
author visited Microsoft Research Silicon Valley. 
}
\and Kunal Talwar\thanks{Microsoft Research Silicon Valley, \texttt{kunal@microsoft.com}.}}

\maketitle
\thispagestyle{empty}

\begin{abstract}
We consider the noise complexity of differentially private mechanisms in the
setting where the user asks $d$ linear queries $f\colon\Rn\to\Re$
non-adaptively. Here, the database is represented by a vector in $\Rn$ and
proximity between databases is measured in the $\ell_1$-metric.

We show that the noise complexity is determined by two geometric
parameters associated with the set of queries.
We use this connection to give tight upper and lower bounds on the noise
complexity for any $d \leq n$. We show that for $d$ random linear queries of
sensitivity~$1$, it is necessary and sufficient to add $\ell_2$-error
$\Theta(\min\{d\sqrt{d}/\epsilon,d\sqrt{\log (n/d)}/\epsilon\})$ to achieve
$\epsilon$-differential privacy.  Assuming the truth of a deep conjecture from
convex geometry, known as the Hyperplane conjecture, we can extend our results to arbitrary linear queries giving nearly matching upper and lower bounds.

Our bound translates to error $O(\min\{d/\epsilon,\sqrt{d\log
(n/d)}/\epsilon\})$ per answer. The best previous upper bound (Laplacian
mechanism) gives a bound of $O(\min\{d/\eps,\sqrt{n}/\epsilon\})$ per answer,
while the best known lower bound was $\Omega(\sqrt{d}/\epsilon)$.
In contrast, our lower bound is strong enough to separate the concept of
differential privacy from the notion of approximate differential
privacy where an upper bound of $O(\sqrt{d}/\epsilon)$ can be achieved.
\end{abstract}

\vfill
\pagebreak
\section{Introduction}

The problem of Privacy-preserving data analysis has attracted a lot of
attention in recent years. Several databases, e.g. those held by the Census
Bureau, contain private data provided by individuals, and protecting the
privacy of those individuals is an important concern. {\em Differential Privacy}
is a rigorous notion of privacy that allows statistical analysis of sensitive
data while providing strong privacy guarantees even in the presence of an
adversary armed with arbitrary auxiliary information.
We refer the reader to the survey of Dwork~\cite{Dwork08} and
the references therein for further motivation and background information.

We consider the following general setting: A \emph{database} is represented by a 
vector $x\in \Re^n$. The \emph{queries} that the analyst may ask are \emph{linear}
combinations of the entries of $x$. More precisely, a multidimensional query is
a map $F\colon \Re^n \rightarrow \Re^d$, and we will restrict ourselves to
linear maps $F$ with coefficients in the interval $[-1,1]$.
Thus $F$ is a $d\times n$ matrix with entries in $[-1,1]$. In this work, we
assume throughout that $d \leq n$. A mechanism is a randomized algorithm which
holds a database $x\in\Rn$, receives a query $F\colon\Rn\to\Rd$ and answers
with some $a\in\Rd.$ Informally, we say a mechanism satisfies
differential privacy in this setting if the densities of the output distributions
on inputs $x, x' \in \Re^n$ with $\|x-x'\|_1\leq 1$ are point wise within
an $\exp(\eps)$ multiplicative factor of each other.
Here and in the following, $\eps>0$ is a parameter that
measures the strength of the privacy guarantee (smaller $\eps$ being a
stronger guarantee). The \emph{error} of a mechanism is the expected Euclidean
distance between the correct answer $Fx$ and the actual answer $a.$

In this work, we use methods from convex geometry to determine a nearly optimal
trade-off between privacy and error. We will see a lower bound on how much
error any differentially private mechanism must add. And we present a
mechanism whose error nearly matches this lower bound.

As mentioned, the above setup is fairly general.
To illustrate it and facilitate comparison
with previous work, we will describe some specific instantiations below.

\paragraph{Histograms.}
Suppose we have a database $y \in [n]^N$, containing
private information about $N$ individuals.
We can think of each individual as belonging to one
of $n$ types. The database $y$ can then naturally be translated to
a histogram $x \in \Re^n$, i.e., $x_i$ counts the number of individuals
of type $i$. Note that in the definition of differential privacy,
we require the mechanism to be defined for all $x \in \Re^n$
and demand that the output distributions be close whenever
$\|x-x'\|_1 \leq 1$. This is a stronger requirement than asserting
this property only for integer vectors $x$ and $x'$. It only makes our
upper bounds stronger. For the lower bounds, this strengthening allows us to ignore the discretization issues that would arise in the usual definition. However, our lower bounds can be extended for the usual
definition for small enough~$\eps$ and large enough $N$ 
(see Appendix~\ref{sec:lbhamming}). 
Now, our upper bound holds for any linear query on the histogram.
This includes some well-studied and natural classes of
queries. For instance, \emph{contingency tables} (see, e.g.,
\cite{BarakCDKMT07}) are linear queries on the histogram.

\paragraph{Private bits.}
In the setting looked at by Dinur and Nissim~\cite{DinurN03}, the database $y
\in \{0,1\}^N$ consists of one private bit for each individual and each
query ask for the number of $1$'s amongst a (random) subset on $[N]$. Given $d$
such queries, one can define $n \leq 2^d$ types of individuals, depending on the
subset of the queries that ask about an individual. The vector $y$ then maps
to a histogram $x$ in the natural way with $x_i$ denoting the number of
individuals of type $i$ with their private bit set to $1$. Our results then
imply a lower bound of $\Omega(d/\eps)$ per answer for any
$\eps$-differentially private mechanism. This improves on the
$\Omega(\sqrt{d})$ bound for $d=N$ from~\cite{DinurN03} for a weaker privacy
definition (blatant non-privacy). A closely related rephrasing is to imagine each individual having $d$
private $\{0,1\}$ attributes so that $n=2^d$. The $d$ queries that ask for the
$1$-way marginals of the input naturally map to a matrix $F$ and
Theorem~\ref{thm:volume-lower} implies a lower bound of $\Omega(d/\eps)$
noise per marginal for such queries.

One can also look at $x$ itself as a database where each individuals
private data is in $[0,1]$; in this setting the dimension of the data $n$
equals the number of individuals $N$. Our results lead to better upper bounds
for this setting.

Finally, there are settings such as the recent work of~\cite{McSherryM09} on private recommendation systems, where the private data is transformed with a stability guarantee so that nearby databases get mapped to vectors at
$\ell_1$ distance at most~$1$.

\subsection{Our results}
We relate the noise complexity of differentially private mechanisms to
some geometric properties of the image of the unit $\ell_1$-ball, denoted
$B_1^n$, when applying the linear mapping $F$. We will denote the resulting 
convex polytope by $K=FB_1^n.$ Our first result lower bounds the noise any
$\epsilon$-differentially private mechanism must add in terms of the volume
of~$K$.
\begin{theorem}\label{thm:volume-lower}
Let $\eps>0$ and suppose $F\colon\Rn\to\Rd$ is a linear map.
Then, every $\eps$-private mechanism $M$ has error at least
$\Omega(\eps^{-1}d\sqrt{d}\cdot\vol(K)^{1/d})$
where $K=FB_1^n$.
\end{theorem}
Recall, the term \emph{error} refers to the expected Euclidean distance between
the output of the mechanism and the correct answer to the query $F$.

We then describe a differentially private mechanism whose error depends on the
expected $\ell_2$ norm of a randomly chosen point in $K$. Our mechanism is an
instantiation of the exponential mechanism~\cite{McSherryTa07} with the score
function defined by the (negative of the) norm $\|\cdot\|_K$, that is the
norm which has $K$ as its unit ball. Hence, we will refer to this mechanism as 
the $K$-norm mechanism. Note that as the definition of this norm
depends on the query~$F,$ so does the output of our mechanism. 

\begin{theorem}\label{thm:upper1}
Let $\eps>0$ and suppose $F\colon\Rn\to\Rd$ is a linear map with $K=FB_1^n.$
Then, the $K$-norm mechanism is $\eps$-differentially private and has error
at most
$O(\eps^{-1}d\E_{z\in K}\|z\|_2).$
\end{theorem}
As it turns out, when $F$ is a random Bernoulli $\pm1$ matrix our upper bound
matches the lower bound up to constant factors. In this case, $K$ is a random
polytope and its volume and average Euclidean norm have been determined rather
recently. Specifically, we apply a volume lower bound of Litvak et
al.~\cite{LitvakPaRuJa05}, and an upper bound on the average Euclidean norm due
to Klartag and Kozma~\cite{KlartagKo09}. Quantitatively,
we obtain the following theorem.

\begin{theorem}\label{thm:random1}
Let $\eps>0$ and $d\le n/2.$ Then, for almost all matrices $F\in\{-1,1\}^{d\times n}$,
\begin{enumerate}
\item
any $\eps$-differentially private mechanism $M$ has error
$\Omega(d/\eps)\cdot\min\{\sqrt{d},\sqrt{\log(n/d)}\}$.
\item
the $K$-norm mechanism is $\eps$-differentially private
with error
$O(d/\eps)\cdot\min\{\sqrt{d},\sqrt{\log(n/d)}\}.$
\end{enumerate}
\end{theorem}
We remark that Litvak et al.~also give an explicit construction of a mapping~$F$
realizing the lower bound.

More generally, we can relate our upper and lower bounds whenever the body $K$
is in \emph{approximately isotropic position}. Informally, this condition implies that
$\E_{z \in K}\|z\|\sim\sqrt{d}\cdot\vol(K)^{1/d}L_K.$ Here, $L_K$ denotes the
so-called \emph{isotropic constant} which is defined in Section~\ref{sec:isotropic}.
\begin{theorem}
Let $\eps>0$ and suppose $F\colon\Rn\to\Rd$ is a linear map such that
$K=FB_1^n$ is in approximately isotropic position.
Then, the $K$-norm mechanism is $\eps$-differentially private with error at most
$O(\eps^{-1}d\sqrt{d}\cdot \vol(K)^{1/d} L_K),$
where $L_K$ denotes the isotropic constant of $K$.
\end{theorem}
Notice that the bound in the previous theorem differs from the lower bound by
a factor of $L_K.$ A central conjecture in convex geometry, sometimes
referred to as the ``Hyperplane Conjecture'' or ``Slicing Conjecture''
(see~\cite{KlartagKo09} for further information) states that $L_K=O(1).$

Unfortunately, in general the polytope $K$ could be very far from isotropic.
In this case, both our volume-based lower bound and the $K$-norm mechanism can be quite far from optimal. We
give a recursive variant of our mechanism and a natural generalization of our volume-based lower bound which are nearly optimal even if $K$ is non-isotropic.

\begin{theorem}\label{thm:non-isotropic}
Let $\eps>0$. Suppose $F\colon\Rn\to\Rd$ is a linear map. Further,
assume the Hyperplane Conjecture. Then, the mechanism introduced in
Section~\ref{sec:nonisotropic} is
$\eps$-differentially private and has error at most
$O(\log^{3/2} d) \cdot \Gvollb(K,\eps).$
where $\Gvollb(K,\eps)$ is a lower bound on the error of the optimal $\eps$-differentially private mechanism.
\end{theorem}


While we restricted our theorems to $F\in[-1,1]^{d\times n}$,
they apply more generally to any linear mapping $F.$

\paragraph{Efficient Mechanisms.}
Our mechanism is an instantiation of the exponential mechanism and involves
sampling random points from rather general high-dimensional convex bodies.
This is why our mechanism is not efficient as it is. However, we can use
rapidly mixing geometric random walks for the sampling step.
These random walks turn out to approach the uniform distribution in a metric
that is strong enough for our purposes. It will follow that
both of our mechanisms can be implemented in polynomial time.

\begin{theorem}
The mechanisms given in Theorem~\ref{thm:upper1} and
Theorem~\ref{thm:non-isotropic} can be implemented in time polynomial
in $n,1/\eps$ such that
the stated error bound remains the same up to constant factors,
and the mechanism achieves $\eps$-differential privacy.
\end{theorem}

%
%

We note that our lower bound $\Gvollb$ can also be approximated up to a 
constant factor. Together these results give
polynomial time computable upper and lower bounds on the error of
any differentially private mechanism, that are always within
an $O(\log^{3/2} d)$ of each other.

\begin{figure}
\begin{center}
\vspace{-3mm}
\begin{tabular}{|l|c|c|c|}
\hline
{\bf Mechanism} & {\bf $\ell_2$-error} & {\bf privacy} & reference \\
\hline
\hline
Laplacian noise & $\eps^{-1}d\sqrt{d}$ & $\eps$ & ~\cite{DworkMNS06}\\
\hline
$K$-norm & $\eps^{-1}d\sqrt{\log(n/d)}$ & $\eps$ & this paper \\
\hline
\hline
lower bound & $\Omega(\eps^{-1}d)$ & $(\eps,\delta)$ & ~\cite{DinurN03}\\
\hline
lower bound & $\Omega(\eps^{-1}d)\min\{\sqrt{\log(n/d},\sqrt{d}\}$ &
$\eps$ & this paper \\
\hline
\end{tabular}
\end{center}
\vspace{-5mm}
\label{fig:summary}
\caption{Summary of results in comparison to best previous work 
for $d$ \emph{random} linear queries each of
sensitivity~$1$ where $1\le d\le n.$ \footnotesize{Note that informally the 
average per coordinate error is smaller 
than the stated bounds by a factor of $\sqrt{d}.$ Here,
$(\eps,\delta)$-differential privacy refers to a weaker approximate notion of
pricacy introduced later. Our lower bound does not apply to this notion.
}}
\end{figure}

Figure~\ref{fig:summary} summarizes our results. Note that we state our bounds
in terms of the total $\ell_2$ error, which informally is a $\sqrt{d}$ factor 
larger than the average per coordinate error.
\subsection{Previous Work}

Queries of the kind described above have (total) \emph{sensitivity} $d$, and 
hence the work of Dwork et al.~\cite{DworkMNS06} shows that adding Laplace noise with
parameter $d/\eps$
to each entry of $Fx$ ensures $\eps$-differential privacy. Moreover, adding
Laplace noise to the histogram $x$ itself leads to another private mechanism.
Thus such questions can be answered with noise $\min(d/\eps,\sqrt{n}/\eps,N)$
per entry of $Fx$. Some specific classes of queries can be answered with
smaller error. Nissim, Raskhodnikova and Smith~\cite{NissimRS07} show that one
can add noise proportional to a smoothed version of the {\em local
sensitivity} of the query, which can be much smaller than the global
sensitivity for some \emph{non-linear} queries. 
Blum, Ligett and
Roth~\cite{BlumLR08} show that it is possible to release approximate counts
for all concepts in a concept class $C$ on $\{0,1\}^m$ with error $O((N^2 m
{\rm VCDim}(C)/\eps)^{\frac{1}{3}})$, where ${\rm VCDim}(C)$ is the VC
dimension of the concept class. Their bounds are incomparable to ours, and in
particular their improvements over the Laplacian mechanism kick in 
when the number of queries is larger than the size of the database (a
range of parameters we do not consider).
Feldman et al.~\cite{FeldmanFKN09} construct private core sets for the $k$-median
problem, enabling approximate computation of the $k$-median cost of any set of
$k$ facilities in $\Re^d$.  Private mechanisms with small error, for other
classes of queries have also been studied in several other works, see
e.g.~\cite{BlumDMN05,BarakCDKMT07,McSherryTa07,ChaudhuriM08,GuptaLMRT10}.

Dinur and Nissim~\cite{DinurN03} initiated the study of  lower bounds on the
amount of noise private mechanisms must add. They showed that any private
mechanism that answers $\tilde{O}(N)$ random subset sum queries about a set of
$N$ people each having a private bit must add noise $\Omega(\sqrt{N})$ to
avoid nearly full disclosure of the database ({\em blatant non-privacy}). This
implies that as one answers more and more questions, the amount of error
needed per answer must grow to provide any kind of privacy guarantee. These
results were strengthened by Dwork, McSherry and Talwar~\cite{DworkMT07}, and
by Dwork and Yekhanin~\cite{DworkY08}. However all these lower bounds protect
against blatant non-privacy and cannot go beyond noise larger than $\min(\sqrt{d},\sqrt{N})$
per answer, for $d$ queries. Kasiviswanathan, Rudelson and Smith~\cite{KasiviswanathanRS09} show lower
bounds of the same nature ($\min(\sqrt{d},\sqrt{N})$ for $d$ questions) for a more natural and useful class of questions. Their lower bounds also apply to $(\eps,\delta)$-differential privacy and are
tight when $\eps$ and $\delta$ are constant. For the case of $d=1$, Ghosh, Roughgarden and Sundararajan~\cite{GhoshRS09} show that adding Laplace noise is in fact optimal in a very general decision-theoretic framework, for any symmetric decreasing loss function. For the case that all sum queries need to be answered (i.e. all queries of the form $f_P(y) = \sum_{i=1}^N P(y_i)$ where $P$ is a $0$-$1$ predicate), Dwork et al.~\cite{DworkMNS06} show that any differentially private mechanism must add noise $\Omega(N)$. Rastogi et al.~\cite{RastogiHS07} show that half of such queries must have error $\Omega(\sqrt{N})$.  Blum, Ligett and Roth~\cite{BlumLR08} show that any differentially private mechanism answering all (real-valued) halfspace queries must add noise $\Omega(N)$.

\subsection{Overview and organization of the paper}
In this section we will give a broad overview of our proof and
outline the remainder of the paper.

Section~\ref{sec:prelim} contains some preliminary facts and definitions. 
Specifically, we describe  
a linear program that defines the optimal mechanism for any set of queries. This
linear program (also studied in~\cite{GhoshRS09} for the one-dimensional case)
is exponential in size, but in principle, given any query and error function,
can be used to compute the best mechanism for the given set of queries.
Moreover, dual solutions to this linear program can be used to prove lower
bounds on the error. However, the asymptotic behavior of the optimum value of
these programs for multi-dimensional queries was not understood prior to this
work. Our lower bounds can be reinterpreted as dual solutions to the
linear program. The upper bounds give near optimal primal solutions. Also, our
results lead to a polynomial-time approximation algorithm for the optimum when
$F$ is linear.

We prove our lower bound in Section~\ref{sec:lower}.
Given a query $F\colon\Rd\to\Rd$, our lower bound depends on the
$d$-dimensional volume
of $K=FB_1^n.$ If the volume of $K$ is large, then a packing argument shows
that we can pack exponentially many points inside $K$ so that each pair of
points is far from each other. We then scale up $K$ by a suitable factor
$\lambda.$ By linearity, all points within $\lambda K$ have preimages under
$F$ that are still $\lambda$-close in $\ell_1$-distance. Hence, the definition 
of $\eps$-differential privacy (by transitivity) enforces some constraint 
between these preimages.  
We can combine these observations so as to show that any differentially
private mechanism $M$ will have to put significant probability mass in
exponentially many disjoint balls. This forces the mechanism to have large
expected error.  

We then introduce the $K$-norm mechanism in Section~\ref{sec:knorm}. Our
mechanism computes $Fx$ and then adds a noise vector to $Fx.$ The key point
here is
that the noise vector is not independent of $F$ as in previous works.
Instead, informally speaking, the noise is tailored to the exact shape of 
$K=FB_1^n.$ This is accomplished by picking a particular noise vector $a$ with 
probability
proportional to $\exp(-\eps\|Fx-a\|_K).$ Here, $\|\cdot\|_K$ denotes the
(Minkowski) norm defined by $K$. While our mechanism depends upon the query~$F$,
it does \emph{not} depend on the particular database $x.$
We can analyze our mechanism in terms of the expected Euclidean distance from
the origin of a random point in $K$, i.e., $\E_{z\in K}\|z\|_2.$ Arguing
optimality of our mechanism hence boils down to relating $\E_{z\in K}\|z\|_2$
to the volume of $K$ which is the goal of the next section.

Indeed, using several results from convex geometry, we observe 
that our lower and upper bounds match up to constant factors when $F$ is 
drawn at random from $\{-1,1\}^{d\times n}$.
As it turns out the polytope $K$ can be interpreted as the symmetric convex 
hull of the row
vectors of $F.$ When $F$ is a random matrix, $K$ is a well-studied random
polytope. Some recent results on random polytopes
give us suitable lower bounds on the volume and upper bounds on the average
Euclidean norm.
More generally, 
our bounds are tight whenever
$K$ is in isotropic position (as pointed out in Section~\ref{sec:isotropic}).
This condition intuitively gives a relation
between volume and average distance from the origin. Our bounds are actually
only tight up to a factor of $L_K,$ the isotropic constant of $K.$ A
well-known conjecture from convex geometry, known as the Hyperplane Conjecture
or Slicing Conjecture, implies that $L_K=O(1).$

The problem is that when $F$ is not drawn at random, $K$ could be very far from isotropic. In this case, the $K$-norm mechanism by itself might actually perform
poorly. We thus give a recursive variant of the $K$-norm mechanism in
Section~\ref{sec:nonisotropic} which can handle non-isotropic bodies. Our
approach is based on analyzing the covariance matrix of~$K$ in order to
partition $K$ into parts on which our earlier mechanism performs well.
Assuming the Hyperplane conjecture, we derive bounds on the error of our
mechanism that are optimal to within polylogarithmic factors.

%
The costly step in both of our mechanisms is sampling uniformly from 
high-dimensional convex bodies such as $K=FB_1^n$. 
To implement the sampling step efficiently, we will use
geometric random walks. It can be shown that these random walks 
approach the uniform distribution over $K$ in polynomial time. We will 
actually need convergence bounds in the relative $\ell_\infty$-metric, 
a metric strong enough to entail guarantees about exact
differential privacy rather than approximate differential privacy (to be
introduced later).

Some complications arise, since we need to repeat the privacy and optimality 
analysis of our mechanisms in the presence of approximation errors 
(such as an approximate covariance matrix and an approximate
separation oracle for $K$). The details 
can be found in Section~\ref{sec:efficient}.

%


\paragraph{Acknowledgments.}
We would
like to thank Frank McSherry, Aaron Roth, Katrina Ligett, Indraneel Mukherjee,
Nikhil Srivastava for several useful discussions, and Adam Smith for
discussions and comments on a previous version of the paper.

\onlyproc{

\section{Lower bounds from volume estimates}
\label{Asec:lower}

In this section we show that lower bounds on the volume of the convex body
$FB_1^n\subseteq\Rd$ give rise to lower bounds on the error that any
private mechanism must have with respect to $F$. For that we will need a formal definition of private mechanism and its error.
\begin{definition}
A \emph{mechanism} $M$ is a family of probability measures
$M=\{\mu_x\colon x\in\Rn\}$
where each measure~$\mu_x$ is defined on $\Rd$.
A mechanism is called \emph{$\eps$-differentially private}, if for all $x,y\in\Rn$ such that
$\|x-y\|_1\le1$, we have
$\sup_{S\sse\Rd}\frac{\mu_x(S)}{\mu_y(S)}\le\exp(\eps),$
where the supremum runs over all measurable subsets $S\sse\Rd$.

For a mapping $F\colon\Rn\to\Rd,$ we define the \emph{$\ell_2$-error} of a 
mechanism $M$ as $\err(M,F)=\sup_{x\in\Rn}\E_{a\sim\mu_x}\|a-Fx\|_2.$
\end{definition}

\begin{definition}
A set of points $Y \subseteq \Re^d$ is called a {\em $r$-packing} if $\|y-y'\|_2 \geq r$ for any $y,y' \in Y, y \neq y'$.
\end{definition}

\begin{fact}\label{Afact:balls}
Let $K\subseteq\Rd$ such that $R=\vol(K)^{1/d}$.
Then, $K$ contains an $\Omega(R\sqrt{d})$-packing of size  at least $\exp(d)$.
\end{fact}

\begin{proof}
Since $\vol(B_2^d)^{1/d}\sim\frac1{\sqrt{d}}$, the body $K$ has the volume of
a ball of radius $r \in \Omega(R\sqrt{d})$. Any maximal $\frac{r}{4}$-packing then has the desired property.
\end{proof}

\begin{theorem}\label{Athm:volume}
Let $\eps>0$ and suppose $F\colon\Rn\to\Rd$ is a linear map and let $K=FB_1^n$.
Then, every
$\eps$-differentially private mechanism $M$
must have
$\err(M,F) \ge \Omega(\eps^{-1}d\sqrt{d}\cdot\vol(K)^{1/d}).$
\end{theorem}

\begin{proof}
Let $\lambda\ge1$ be some scalar and put $R=\vol(K)^{1/d}$.
By Fact~\ref{Afact:balls} and our assumption, $\lambda K=\lambda FB_1^n$
contains an $\Omega(\lambda R\sqrt{d})$-packing $Y$ of size at least $\exp(d)$. Let $X\subseteq \Rn$ be a set of arbitrarily chosen preimages of $y \in Y$ so that $|X|=|Y|$ and $FX = Y$.
By linearity, $\lambda FB_1^n=F(\lambda B_1^n)$ and hence we may assume that
every $x\in X$ satisfies $\|x\|_1\le \lambda$.

We will now assume that $M=\{\mu_x\colon x\in\Rn\}$ is an $\eps$-differentially private
mechanism with error $cd\sqrt{d}R/\eps$ and lead this to a
contradiction for small enough $c>0$. For this we set $\lambda=d/2\eps$.
By the assumption on the error, Markov's inequality implies that for all $x\in X$,
we have $\mu_x(B_x)\ge\tfrac12,$
where $B_x$ is a ball of radius $2cd\sqrt{d}R/\eps = 4c\lambda R\sqrt{d}$ centered at $Fx$. Since $Y=FX$ is an $\Omega(\lambda R\sqrt{d})$-packing, the balls $\{B_x: x\in X\}$ are disjoint for small enough constant $c>0$.

Since $\|x\|_1\le \lambda$, it follows from the transitivity of 
$\e$-differential privacy (shown in Fact~\ref{fact:trans}) that
$\mu_0(B_x)
\ge\exp(-\eps \lambda)\mu_{x}(B_x)\ge\tfrac12\exp(-d/2).$
Since the balls $B_x$ are pairwise disjoint,
\begin{equation}
1\ge \mu_0(\cup_{x\in X}B_x)
=\sum_{x\in X}\mu_0(B_x)
\ge \exp(d)\tfrac12\exp(-d/2)
> 1
\end{equation}
for $d\geq 2$. We have thus obtained a contradiction.
\end{proof}

We denote by $\vollb(F,\eps))$ the lower bound resulting from the above theorem. In other words
\[
\vollb(F,\eps) = \eps^{-1} d\sqrt{d}\cdot\vol(FB_1^n)^{1/d}.
\]
Thus any $\eps$-differentially private mechanism must add noise
$\Omega(\vollb(F,\eps))$.

\section{The $K$-norm mechanism}
\label{Asec:knorm}

In this section we describe a new differentially private mechanism,
which we call the $K$-norm mechanism.

\begin{definition}[$K$-norm mechanism]
Given a linear map $F\colon\Rn\to\Rd$ and $\eps>0$, we let $K=FB_1^n$ and define
the mechanism ${\bf KM}(F,d,\eps)=\{\mu_x\colon x\in\Rn\}$ so that each measure
$\mu_x$ is given by the probability density function
\begin{equation}\label{eq:density}
f(a)=Z^{-1}\exp(-\eps\|Fx-a\|_K)
\end{equation}
defined over $\Rd.$ Here $Z$ denotes a normalization constant independent of $x.$
\end{definition}
A more concrete view of the mechanism is provided by Figure~\ref{Afig:km}. In
Remark~\ref{rem:gamma} we justify this description by showing that it is
equivalent to the above definition.
%

\begin{figure}
\begin{center}
\fbox{
\begin{minipage}{.9\textwidth}
${\bf KM}(F,d,\eps)\colon$
\vspace{-3mm}
\begin{enumerate*}
\item Sample $z$ uniformly at random from $K=FB_1^n$ and sample
$r\sim\GammaDist(d+1,\eps^{-1}).$
\item Output $Fx+rz.$
\end{enumerate*}
\end{minipage}
}
\end{center}
\vspace{-5mm}
\caption{Description of the $d$-dimensional $K$-norm mechanism.
\footnotesize{Here, $\GammaDist(d+1,\eps^{-1})$ denotes the
$(d+1)$-dimensional Gamma distribution with scaling $\eps^{-1}.$}}
\label{Afig:km}
\end{figure}

The next theorem shows that the $K$-norm mechanism is indeed differentially
private. Moreover, we can express its error in terms of the \emph{expected
distance from the origin} of a random point in $K.$

\begin{theorem}\label{Athm:knorm}
Let $\eps>0$. Suppose $F\colon\Rn\to\Rd$ is a linear map and put $K=FB_1^n.$
Then, the mechanism ${\rm KM}(F,d,\eps)$ is
$\eps$-differentially private, and for every
$p>0$ achieves the error bound 
$\E_{a\sim \mu_x} \|Fx-a\| \leq 
\frac{d+1}{\eps}\E_{z\in K}\|z\|_2.$
\end{theorem}

\begin{proof}
To argue the error bound, we follow the description in Figure~\ref{Afig:km}. 
Let $D=\GammaDist(d+1,1/\eps).$ For all~$x\in\Rn$,
\begin{equation*}
\E_{a\sim \mu_x}\|Fx-a\|
= \E_{a\sim \mu_0} \|a\|
= \E_{r\sim D} \E_{a\in rK} \|a\|
= \E_{r\sim D} r \cdot \E_{z\in K} \|z\|
= \frac{\Gamma(d+2)}{\eps\Gamma(d+1)}\E_{z\in K}\|z\|\mper 
\end{equation*}
In the first step we used that noise of the mechanism does not depend on $x.$ 
In the second to last step we applied a formula for the first moment of the
Gamma distribution (see Fact~\ref{fact:gamma}).
Moreover, $\frac{\Gamma(d+2)}{\Gamma(d+1)}=d+1.$

Privacy follows from the fact that the mechanism is a special case of the
exponential mechanism~\cite{McSherryTa07}. For completeness, we repeat
the argument.

Suppose that $\|x\|_1\le1$. It suffices to show that for all $a\in\Rd$, the
densities of $\mu_0$ and $\mu_x$ are within multiplicative $\exp(\eps)$, i.e.,
\begin{align*}
\frac{Z^{-1}e^{-\eps\|a\|_K}}
{Z^{-1}e^{-\eps\|Fx-a\|_K}}
 = e^{\eps(\|Fx-a\|_K - \|a\|_K)}
 \le e^{\eps\|Fx\|_K}
 \le e^\eps.
\end{align*}
where in the first inequality we used the triangle inequality for
$\|\cdot\|_K$. In the second step we used
that $x\in B_1^n$ and hence $Fx\in FB_1^n=K$ which means $\|Fx\|_K\le1.$
Hence, the mechanism satisfies $\eps$-differential privacy.
\end{proof}

\section{Optimality for random queries and isotropic bodies}
\label{Asec:random}

In this section, we will show that our upper bound matches our lower bound
when~$F$ is a random query.
A key observation is that $FB_1^n$ is the \emph{symmetric} convex hull
of $n$ (random) points $\{v_1,\dots,v_n\}\sse\Rd$, i.e., the convex hull
of $\{\pm v_1,\dots,\pm v_n\}$, where $v_i \in \Rd$ is the $i$th column of $F$.
The symmetric convex hull of random points has been
studied extensively in the theory of random polytopes.
A recent result of Litvak, Pajor, Rudelson and
Tomczak-Jaegermann~\cite{LitvakPaRuJa05} gives
the following lower bound on the volume of the convex hull.

\begin{theorem}[\cite{LitvakPaRuJa05}]
Let $2d\le n\le 2^d$ and let $F$ denote a random $d\times n$ Bernoulli matrix.
Then,
\begin{equation}\label{Aeq:vol}
\textstyle
\vol(FB_1^n)^{1/d}\ge \Omega(1)\sqrt{\log(n/d)/d}\mcom
\end{equation}
with probability
$1-\exp(-\Omega(d^\beta n^{1-\beta}))$ for any $\beta\in(0,\frac12).$
Furthermore, there is an explicit construction of $n$ points in $\{-1,1\}^d$
whose convex hull achieves the same volume.
\end{theorem}

The bound in~(\ref{Aeq:vol}) is tight up to constant factors.
A well known result~\cite{BaranyFu88} shows that
the volume of the convex hull of any $n$ points on the sphere in $\Rd$
of radius $\sqrt{d}$ is bounded by
%
\begin{equation}
\vol(K)^{1/d}\le O(1)\sqrt{\log(n/d)/d}\mper
\end{equation}
%
Notice, that in our case
$K=FB_1^n\sse B_\infty^d\sse\sqrt{d}B_2^d$
and in fact the vertices of $K$ are points on the $(d-1)$-dimensional sphere
of radius $\sqrt{d}$. However, equation~(\ref{Aeq:vol}) states that the
normalized volume of the random polytope $K$ will be proportional to the
volume of the Euclidean ball of radius $\sqrt{\log(n/d)}$ rather than
$\sqrt{d}.$
When $d\gg\log n$, this means that the volume of $K$ will be tiny compared to
the volume of the infinity ball~$B_\infty^d$.
Note that~(\ref{Aeq:vol}) does not apply when $d\le \log n.$ But in that case
it is easy to get a lower bound as follows: Simply consider $F$ corresponding
to all $n$ points of the boolean hypercube $\{-1,1\}^d.$ Now,
$FB_1^n\supseteq[-1,1]^d$ and hence $\vol(FB_1^n)^{1/d}\ge 2.$ This explains
why we get qualitatively different answers below and above $d=\log(n).$

%
%
%
By combining the volume lower bound with Theorem~\ref{Athm:volume},
we get the following lower bound on the error of private mechanisms.
\begin{theorem}\label{Athm:lower}
Let $\eps>0$ and $0<d\le n/2$. Then, for almost all
matrices $F\in\{-1,1\}^{d\times n}$, every
$\eps$-differentially private mechanism $M$
must have
$\err(M,F)\ge\Omega(d/\eps)\cdot\min\left\{\sqrt{d},\sqrt{\log(n/d)}\right\}.$
\end{theorem}

\paragraph{A separation result.}
We use this paragraph to point out that our lower bound immediately implies a
separation between approximate differential privacy
(see Definition~\ref{def:epsdelta})
and exact differential privacy (as used throughout the paper).
The Gaussian mechanism (see Theorem~\ref{thm:gaussian}) gives a mechanism 
providing $\delta$-approximate
$\eps$-differential privacy with error $o(\eps^{-1}\sqrt{\log (n/d)})$ as long
as $\delta\ge 1/n^{o(1)}.$ Our lower bound  in Theorem~\ref{Athm:lower}
on the other hand states that the error of any $\eps$-differentially private
mechanism must be $\Omega(\eps^{-1}\sqrt{\log(n/d)})$
(assuming $d\gg\log(n)$). We get the strongest separation when $d\le\log(n)$ and
$\delta$ is constant. In this case, our lower bound is a factor $\sqrt{d}$
larger than the upper bound for approximate differential privacy.

\subsection{Upper bound on average Euclidean norm}

Klartag and Kozma~\cite{KlartagKo09}
recently gave a bound on the quantity $\E_{z\sim K}\|z\|$ when
$K=FB_1^n$ for random~$F.$

\begin{theorem}[\cite{KlartagKo09}]
Let $F$ be a random $d\times n$ Bernoulli matrix and put $K=FB_1^n$. Then,
there is a constant $C>0$ so that with probability greater than $1-Ce^{-O(n)}$,
$\E_{z\in K}\|z\|^2\le C\log(n/d).$
\end{theorem}

An application of Jensen's inequality thus gives us the following corollary.
\begin{corollary}
Let $\eps>0$ and $0<d\le n/2$. Then, for almost all
matrices $F\in\{-1,1\}^{d\times n}$, the mechanism ${\rm KM}(F,d,\eps)$ is
$\eps$-differentially private with error at most
$O(d/\eps)\cdot\min\left\{\sqrt{d},\sqrt{\log(n/d)}\right\}.$
\end{corollary}

\subsection{Approximately isotropic bodies}
\label{Asec:isotropic}

The following definition is a relaxation of nearly isotropic position used
in literature (e.g.,~\cite{KannanLoSi97})
\begin{definition}[Approximately Isotropic Position]
We say a convex body $K \sse \Re^d$ is in $c$-\emph{approximately isotropic
position} if for every unit vector $v \in \Re^d$,
$\E_{z\in K}|\langle z,v\rangle|^2
\le
c^2L_K^2 \vol(K)^{\frac{2}{d}} \mper$
\end{definition}

The results of Klartag and Kozma~\cite{KlartagKo09} referred to in the
previous section show that the symmetric convex hull $n$ random points from
the $d$-dimensional hypercube are in $O(1)$-approximately isotropic position and have
$L_K = O(1)$. More generally, the $K$-norm mechanism can be shown to be
approximately optimal whenever $K$ is nearly isotropic.

\begin{theorem}[Theorem~\ref{thm:upper1} restated]\label{thm:upper}
Let $\eps>0$. Suppose $F\colon\Rn\to\Rd$ is a linear map such that $K=FB_1^n$
is in $c$-approximately isotropic position. Then, the $K$-norm mechanism is
$\eps$-differentially private and has error at most
$O(c L_K)\cdot \vollb(F,\eps).$
\end{theorem}


We can see that the previous upper bound is tight up to a factor of~$cL_K$.
Estimating $L_K$ for general convex bodies is a well-known open problem in
convex geometry.
The best known upper bound for a general convex body
$K\subseteq \Rd$ is $L_K\le O(d^{1/4})$ due to Klartag~\cite{Klartag06},
improving over the estimate $L_K\le
O(d^{1/4}\log d)$ of Bourgain from '91.
The conjecture is that $L_K=O(1)$.

\begin{conjecture}[Hyperplane Conjecture]
\label{Aconj:hyperplane}
There exists $C>0$ such that for every $d$ and every convex set
$K\sse\Rd$, $L_K<C$.
\end{conjecture}

Assuming this conjecture we get matching bounds for approximately isotropic
convex bodies.
\begin{theorem}
Let $\eps>0.$ Assuming the hyperplane conjecture, for every
$F\in[-1,1]^{d\times n}$ such that $K=FB_1^n$ is $c$-approximately isotropic,
the $K$-norm mechanism ${\rm KM}(F,d,\eps)$ is  $\eps$-differentially private with
error at most
$O(c)\cdot \vollb(F,\eps) \leq
O(cd/\eps)\cdot\min\left\{\sqrt{d},\sqrt{\log(n/d)}\right\}\mper$
\end{theorem}

\section{Non-isotropic bodies}
\label{Asec:nonisotropic}

While the mechanism of the previous sections is near-optimal for
near-isotropic queries, it can be far from optimal if $K$ is far from
isotropic. For example, suppose the matrix $F$ has random entries from
$\{+1,-1\}$ in the first row, and (say) from $\{\frac{1}{d^2},-\frac{1}{d^2}\}$ in the
remaining rows. While the Laplacian mechanism will add $O(\frac{1}{\epsilon})$
noise to the first co-ordinate of $Fx$, the $K$-norm mechanism will add noise
$O(d/\epsilon)$ to the first co-ordinate. Moreover, the volume lower
bound $\vollb$ is at most $O(\eps^{-1}\sqrt{d})$. Rotating $F$ by a random
rotation gives, w.h.p., a query for which the Laplacian mechanism adds
$\ell_2$ error $O(d/\eps)$. For such a body, the Laplacian and the $K$-norm
mechanisms, as well as the $\vollb$ are far from optimal.

In this section, we will design a recursive mechanism that can handle
such non-isotropic convex bodies. To this end, we will need to introduce a few more
notions from convex geometry.

Suppose $K\subseteq\Rd$ is a centered convex body, i.e. $\int_K
x\rd x=0.$
The \emph{covariance matrix of $K$}, denoted $M_K$ is the $d\times d$ matrix
with entry $ij$ equal to
$M_{ij} = \frac1{\vol(K)}\int_{K} x_ix_j\rd x.$ That is, $M_K$ is the
covariance matrix of the uniform distribution over $K.$
%

\paragraph{A recursive mechanism.}
Having defined the covariance matrix, we can
describe a recursive mechanism for the case when $K$ is not in
isotropic position. The idea of the mechanism is to act differently on
different eigenspaces of the covariance matrix. Specifically, the mechanism
will use a lower-dimensional version of ${\bf KM}(F,d',\eps)$ on subspaces
corresponding to few large eigenvalues.

Our mechanism, called ${\rm NIM}(F,d,\eps)$, is given a linear mapping
$F\colon\Rn\to\Rd,$ and parameters $d\in\mathbb{N},\eps>0.$
The mechanism proceeds recursively by partitioning the convex body $K$ into two
parts defined by the middle eigenvalue of $M_K.$ On one part it will act
according to the $K$-norm mechanism. On the other part, it will descend
recursively. The mechanism is described in Figure~\ref{Afig:nim}.

\begin{figure}
\begin{center}
\fbox{
\begin{minipage}{.9\textwidth}
${\bf NIM}(F,d,\eps)\colon$
\vspace{-3mm}
\begin{enumerate*}
\item Let $K=FB_1^n$. Let $\sigma_1\ge\sigma_2\ge\dots\ge\sigma_d$
denote the eigenvalues of the covariance matrix $M_K.$ Pick a corresponding
orthonormal eigenbasis $u_1,\dots,u_d$.
\item
\label{step:subspaces}
Let $d'=\lfloor d/2\rfloor$ and let $U={\rm span}\{u_1,\dots,u_{d'}\}$ and
$V={\rm span}\{u_{d'+1},\dots,v_d\}.$
\item \label{step:knorm}
Sample $a\sim{\bf KM}(F,d,\eps)\mper$
\item If $d=1$, output $P_V a$. Otherwise, output
${\bf NIM}(P_U F,d',\eps) + P_V a\mper$
\end{enumerate*}
\end{minipage}
}
\end{center}
\vspace{-5mm}
\caption{Mechanism for non-isotropic bodies}
\label{Afig:nim}
\end{figure}

\begin{theorem}\label{Athm:nim}
Let $\eps>0$. Suppose $F\colon\Rn\to\Rd$ is a linear map. Further,
assume the hyperplane conjecture.  Then, the mechanism ${\bf NIM}(F,d,\eps)$ is 
$\eps$-differentially private mechanism~$M$ and achieves error at most
$O(\log(d)^{3/2}\cdot \Gvollb(F,\eps)).$
\end{theorem}
While it is easy to argue privacy (up to losing a $\log(d)$ factor in $\eps$
due to the recursion), the error analysis of our mechanism
requires more work. In particular, it is crucial to understand how the volume
of~$P_U K$ compares to the norm of~$P_Va.$ In Section~\ref{sec:eigenvol} we
give a formula for the volume of $K$ in eigenspaces of the covariance matrix
such as~$U$.
This is then used in Section~\ref{sec:nim} to conclude Theorem~\ref{Athm:nim}.

\section{Efficient implementation}
\label{Asec:efficient}

The costly step in our mechanism is sampling uniformly from high-dimensional
convex bodies such as $K=FB_1^n$. 
To implement the sampling step efficiently, we will use
geometric random walks. It can be shown that these random walks 
approach the uniform distribution over $K$ 
in polynomial time. We will actually need convergence bounds in the relative
$\ell_\infty$-metric, a metric strong enough to entail guarantees about exact
differential privacy. 

 Some complications arise, since we need to repeat the 
privacy and optimality analysis of our mechanisms in the presence of
approximation errors (such as an approximate covariance matrix and an approximate
separation oracle for $K$). The details are omitted from the short abstract
but can be found in Section~\ref{sec:efficient}.

}

\section{Preliminaries}
\label{sec:prelim}
\paragraph{Notation.} We will write $B_p^d$ to denote the unit ball of the $p$-norm in $\Rd$. When
$K\subseteq\Rd$ is a centrally symmetric convex set, we write $\|\cdot\|_K$ for the (Minkowski) norm
defined by $K$ (i.e. $\|x\|_K = \inf \{r \colon x \in rK\}$). The
$\ell_p$-norms are denoted by $\|\cdot\|_p$, but we use $\|\cdot\|$
as a shorthand for the Euclidean norm $\|\cdot\|_2$. Given a function $F : \Re^{d_1} \rightarrow \Re^{d_2}$ and a set $K \in \Re^{d_1}$, $FK$ denotes the
set~$\{F(x): x\in K\}$.

\subsection{Differential Privacy}
\begin{definition}
A \emph{mechanism} $M$ is a family of probability measures
$M=\{\mu_x\colon x\in\Rn\}$
where each measure~$\mu_x$ is defined on $\Rd$.
A mechanism is called \emph{$\eps$-differentially private}, if for all $x,y\in\Rn$ such that
$\|x-y\|_1\le1$, we have
$\sup_{S\sse\Rd}\frac{\mu_x(S)}{\mu_y(S)}\le\exp(\eps),$
where the supremum runs over all measurable subsets $S\sse\Rd$.
\end{definition}

A common weakening of $\eps$-differential privacy is the following notion of \emph{approximate} privacy.
\begin{definition}\label{def:epsdelta}
A mechanism is called $\delta$-approximate \emph{$\eps$-differentially private},
if for all $x,y\in\Rn$ such that
$\mu_x(S)\le\exp(\eps)\mu_y(S)+\delta$
for all measurable subsets $S\sse\Rn$ whenever$\|x-y\|_1\le1$,
\end{definition}

The definition of privacy is transitive in the following sense.

\begin{fact}
\label{fact:trans}
If $M$ is an $\eps$-differentially private mechanism and $x,y\in\Rn$ satisfy $\|x-y\|_1\le
k$, then for measurable $S\sse\Rd$ we have
$\frac{\mu_x(S)}{\mu_y(S)}\le\exp(\eps k).$
\end{fact}

\begin{definition}[Error]
Let $F\colon\Rn\to\Rd$ and $\ell\colon\Rd\times\Rd\to\R^+$.
We define the \emph{$\ell$-error} of a mechanism $M$ as
$\err_\ell(M,F)=\sup_{x\in\Rn}\E_{a\sim\mu_x}\ell(a,Fx).$
Unless otherwise specified, we take $\ell$ to be the Euclidean norm~$\ell_2$.
\end{definition}

\begin{definition}[Sensitivity]
We will consider mappings $F$ which possess the
Lipschitz property,
$\sup_{x\in B_1^n}\|Fx\|_1\le d.$
In this case, we will say that $F$ has \emph{sensitivity}~$d$.
\end{definition}

Our goal is to show trade-offs between privacy and error. The following
standard upper bound, usually called the Laplacian mechanism, is known.

\begin{theorem}[\cite{DworkMNS06}]
\label{thm:laplace}
For any mapping $F\colon\Rn\to\Rd$ of sensitivity~$d$ and any~$\eps>0$,
there exists an $\eps$-differentially private mechanism $M$
with $\err(M,F)=O(d\sqrt{d}/\eps).$
\end{theorem}

When it comes to approximate privacy, the so-called Gaussian mechanism
provides the following guarantee.
\begin{theorem}[\cite{DworkKMMN06}]
\label{thm:gaussian}
Let $\eps,\delta>0.$ Then, for any mapping $F\colon\Rn\to\Rd$ of sensitivity~$d$ there exists a $\delta$-approximate $\eps$-differentially private mechanism  $M$ with
$\err(M,F)=O(d\sqrt{\log(1/\delta)}/\eps).$
\end{theorem}

\subsection{Isotropic Position}

\begin{definition}[Isotropic Position]
We say a convex body $K \sse \Re^d$ is in {\em isotropic position} with
isotropic constant $L_K$ if for every unit vector $v \in \Re^d$,
\begin{equation}
\frac1{\vol(K)}\int_K |\langle z,v\rangle|^2 dz = L_K^2\vol(K)^{2/d}\mper
\end{equation}
\end{definition}

\begin{fact}
For every convex body $K \sse \Re^d$, there is a volume-preserving
linear transformation $T$ such that $TK$ is in isotropic position.
\end{fact}

For an arbitrary convex body $K$, its isotropic constant $L_K$ can then be
defined to be $L_{TK}$ where $T$ brings $L$ to isotropic position. It is known
(e.g.~\cite{MilmanPa89}) that $T$ is unique up to an orthogonal transformation
and thus this is well-defined.

We refer the reader to the paper of Milman and Pajor~\cite{MilmanPa89}, as
well as the extensive survey of Giannopoulos~\cite{Giannopoulos03}
for a proof of this fact and other facts regarding the isotropic constant.

\subsection{Gamma Distribution}

The \emph{Gamma distribution} with shape parameter $k>0$ and scale $\theta>0$,
denoted $\GammaDist(k,\theta)$, is given by the probability density function
\[
f(r;k,\theta)=r^{k-1}\frac{e^{-r/\theta}}{\Gamma(k)\theta^k}\mper
\]
Here, $\Gamma(k)=\int e^{-r}r^{k-1}\rd r$ denotes the Gamma function.
We will need an expression for the moments of the Gamma distribution.
\begin{fact}\label{fact:gamma}
Let $r\sim\GammaDist(k,\theta).$ Then,
\item
\begin{equation}
\Ex{r^m}
=\frac{\theta^m\Gamma(k+m)}{\Gamma(k)}\mper
\end{equation}
\end{fact}

\begin{proof}
\begin{align*}
\Ex{r^m}
=\int_\R
r^{k+m-1}\frac{e^{-r/\theta}}{\Gamma(k)\theta^k}\rd r
&=\frac1{\Gamma(k)\theta^k}\int_\R
(\theta r)^{k+m-1}e^{-r}\rd \theta r\\
&=\frac{\Gamma(k+m)\theta^{k+m}}{\Gamma(k)\theta^k}
=\frac{\Gamma(k+m)\theta^{m}}{\Gamma(k)}
\end{align*}
\end{proof}

\subsection{Linear Programming Characterization}

Suppose that the set of databases is given by some set $\dee$, and let $\dist
: \dee \times \dee\rightarrow \Re_0$ be a distance function on $\dee$. A query
$q$ is specified by an error function $\err : \dee \times \Range \rightarrow
\Re$. For example $\dee$ could be the Hamming cube $\{0,1\}^N$ with $\dist$ being the Hamming distance. Given a query $F \colon \{0,1\}^N
\rightarrow \Re^d$, the error function could be $\err(x,a) =
\|a-F(x)\|_2$ if we wish to compute $F(x)$ up to a small $\ell_2$ error.

A mechanism is specified by a distribution $\mu_x$ on $\Range$ for every
$x\in\dee$. Assume for simplicity that $\dee$ and $\Range$ are both finite.
Thus a mechanism is fully defined by real numbers $\mu(x,a)$, where
$\mu(x,a)$ is the probability that the mechanism outputs answer
$a\in \Range$ on databases $x\in\dee$. The constraints on $\mu$ for an
$\epsilon$-differentially private mechanism are given by
\begin{align*}
\sum_{a\in \Range} \mu(x,a) &= 1 & \forall x \in \dee\\
\mu(x,a) & \geq 0 & \forall x \in \dee, a \in \Range\\
\mu(x,a) & \leq \exp(\epsilon \dist(x,x')) \mu(x',a) & \forall x,x' \in \dee, a \in \Range
\end{align*}

The expected error (under any given prior over databases) is then a linear
function of the variables $\mu(x,a)$ and can be optimized. Similarly, the
worse case (over databases) expected error can be minimized, and we will concentrate on this measure
for the rest of the paper. However these linear programs can be
prohibitive in size. Moreover, it is not a priori clear how one can use this
formulation to understand the asymptotic behavior of the error of the optimum
mechanism.

Our work leads to a constant approximation to the optimum of this linear
program when $F$ is a random in $\{-1,+1\}^{d\times n}$ and an
$O(\log^{3/2} d)$-approximation otherwise.

\section{Lower bounds via volume estimates}
\label{sec:lower}

In this section we show that lower bounds on the volume of the convex body
$FB_1^n\subseteq\Rd$ give rise to lower bounds on the error that any
private mechanism must have with respect to $F$. 

\begin{definition}
A set of points $Y \subseteq \Re^d$ is called a {\em $r$-packing} if $\|y-y'\|_2 \geq r$ for any $y,y' \in Y, y \neq y'$.
\end{definition}

\begin{fact}\label{fact:balls}
Let $K\subseteq\Rd$ such that $R=\vol(K)^{1/d}$.
Then, $K$ contains an $\Omega(R\sqrt{d})$-packing of size  at least $\exp(d)$.
\end{fact}

\begin{proof}
Since $\vol(B_2^d)^{1/d}\sim\frac1{\sqrt{d}}$, the body $K$ has the volume of
a ball of radius $r \in \Omega(R\sqrt{d})$. Any maximal $\frac{r}{4}$-packing then has the desired property.
\end{proof}

\begin{theorem}\label{thm:volume}
Let $\eps>0$ and suppose $F\colon\Rn\to\Rd$ is a linear map and let $K=FB_1^n$.
Then, every
$\eps$-differentially private mechanism $M$
must have
$\err(M,F) \ge \Omega(\eps^{-1}d\sqrt{d}\cdot\vol(K)^{1/d}).$
\end{theorem}

\begin{proof}
Let $\lambda\ge1$ be some scalar and put $R=\vol(K)^{1/d}$.
By Fact~\ref{fact:balls} and our assumption, $\lambda K=\lambda FB_1^n$
contains an $\Omega(\lambda R\sqrt{d})$-packing $Y$ of size at least $\exp(d)$. Let $X\subseteq \Rn$ be a set of arbitrarily chosen preimages of $y \in Y$ so that $|X|=|Y|$ and $FX = Y$.
By linearity, $\lambda FB_1^n=F(\lambda B_1^n)$ and hence we may assume that
every $x\in X$ satisfies $\|x\|_1\le \lambda$.

We will now assume that $M=\{\mu_x\colon x\in\Rn\}$ is an $\eps$-differentially private
mechanism with error $cd\sqrt{d}R/\eps$ and lead this to a
contradiction for small enough $c>0$. For this we set $\lambda=d/2\eps$.
By the assumption on the error, Markov's inequality implies that for all $x\in X$,
we have $\mu_x(B_x)\ge\tfrac12,$
where $B_x$ is a ball of radius $2cd\sqrt{d}R/\eps = 4c\lambda R\sqrt{d}$ centered at $Fx$. Since $Y=FX$ is an $\Omega(\lambda R\sqrt{d})$-packing, the balls $\{B_x: x\in X\}$ are disjoint for small enough constant $c>0$.

Since $\|x\|_1\le \lambda$, it follows from $\e$-differential privacy with
Fact~\ref{fact:trans} that
\[
\mu_0(B_x)
\ge\exp(-\eps \lambda)\mu_{x}(B_x)\ge\tfrac12\exp(-d/2).
\]
Since the balls $B_x$ are pairwise disjoint,
\begin{equation}
1\ge \mu_0(\cup_{x\in X}B_x)
=\sum_{x\in X}\mu_0(B_x)
\ge \exp(d)\tfrac12\exp(-d/2)
> 1
\end{equation}
for $d\geq 2$. We have thus obtained a contradiction.
\end{proof}

We denote by $\vollb(F,\eps))$ the lower bound resulting from the above theorem. In other words
$$ \vollb(F,\eps) = \eps^{-1} d\sqrt{d}\cdot\vol(FB_1^n)^{1/d}.$$
Thus any $\eps$-differentially private mechanism must add noise $\Omega(\vollb(K,\eps))$.
We will later need the following modification of the previous argument which
gives a lower bound in the case where $K$ is close to a lower dimensional
subspace and hence the volume inside this subspace may give a stronger
lower bound.

\begin{corollary}\label{cor:volumesubspace}
Let $\eps>0$ and suppose $F\colon\Rn\to\Rd$ is a linear map and let
$K=FB_1^n$. Furthermore, let $P$ denote the orthogonal projection operator of
a $k$-dimensional subspace of $\Rd$ for some $1\le k\le d.$
Then, every
$\eps$-differentially private mechanism $M$
must have
\begin{equation}
\err(M,F) \ge \Omega(\eps^{-1}k\sqrt{k}\cdot\vol_k(PK)^{1/k}).
\end{equation}
\end{corollary}
\begin{proof}
Note that a differentially private answer $a$ to $F$ can be projected down to a (differentially private) answer $Pa$ to $PF$ and $P$ is norm $1$ operator.
\end{proof}

We will denote by $\Gvollb(F,\eps)$ the best lower bound
obtainable in this manner, i.e.,
\[
\Gvollb(F,\eps) = \sup_{k, P} \eps^{-1}k\sqrt{k}\cdot\vol_k(PFB_1^n)^{1/k}
\]
where the supremum is taken over all $k$ and all $k$-dimensional orthogonal projections $P$.

\paragraph{Lower bounds in the Hamming metric.}
Our lower bound used the fact that the mechanism is defined on all vectors
$x\in\Rd$. In Appendix~\ref{sec:lbhamming}, we show how the lower bound can
be extended when restricting the domain of the mechanism to integer vectors
$x\in[N]^n,$ where distance is measured in the Hamming metric.

\subsection{Lower bounds for small number of queries}
\label{sec:lower-small}
As shown previously, the task of proving lower bounds on the
error of private mechanisms reduces to analyzing the volume of $FB_1^n.$
When $d\le\log n$ this is a straightforward task.

\begin{fact}\label{lem:smalld}
Let $d\le\log n$. Then, for all matrices
$F\in[-1,1]^{d\times n}$,
$\vol(FB_1^n)^{1/d}\le O(1).$
Furthermore, there is an explicit matrix $F$ such that $FB_1^n$ has
maximum volume.
\end{fact}

\begin{proof}
Clearly, $FB_1^n$ is always contained in $B_\infty^d$ and
$\vol(B_\infty^d)^{1/d}=2$.
On the other hand, since $n\ge 2^d$,
we may take $F$ to contain all points of the hypercube
$H=\{\pm 1\}^d$ as its columns.
In this case, $FB_1^n\supseteq B_\infty^d$.
\end{proof}

This lower bound shows that the standard upper bound from
Theorem~\ref{thm:laplace} is, in fact, optimal when $d\le\log(n).$

\section{The $K$-norm mechanism}
\label{sec:knorm}

In this section we describe a new differentially private mechanism,
which we call the $K$-norm mechanism.

\begin{definition}[$K$-norm mechanism]
Given a linear map $F\colon\Rn\to\Rd$ and $\eps>0$, we let $K=FB_1^n$ and define
the mechanism ${\bf KM}(F,d,\eps)=\{\mu_x\colon x\in\Rn\}$ so that each measure
$\mu_x$ is given by the probability density function
\begin{equation}\label{eq:density}
f(a)=Z^{-1}\exp(-\eps\|Fx-a\|_K)
\end{equation}
defined over $\Rd.$ Here $Z$ denotes the normalization constant
\[
Z=\int_{\Rd}\exp(-\eps\|Fx-a\|_K)\rd a= \Gamma(d+1)\vol(\eps^{-1} K).
\]
\end{definition}
A more concrete view of the mechanism is provided by Figure~\ref{fig:km} and
justified in the next remark.
\begin{remark}\label{rem:gamma}
We can sample from the distribution $\mu_x$ as follows:
\begin{enumerate}
\item Sample $r$ from the Gamma distribution with parameter $d+1$ and scale
$\eps^{-1}$, denoted $\GammaDist(d+1,\eps^{-1})$. That is, $r$ is
distributed as
\[
\Pr(r>R)=\frac{1}{\eps^{-d}\Gamma(d+1)}\int_R^\infty e^{-\eps t}t^d\rd t.
\]
\item Sample $a$ uniformly from $Fx+rK$.
\end{enumerate}
Indeed, if $\|a-Fx\|_K=R$, then the distribution of $a$ as above follows the
probability density function
\begin{equation}
\label{eqn:densitycal}
g(a)= \frac{1}{\eps^{-d}\Gamma(d+1)}
\int_R^\infty \frac{e^{-\eps t}t^d}{\vol(tK)}\rd t
=\frac{\int_R^\infty e^{-\eps t}\rd t}{\Gamma(d+1)\vol(\eps^{-1}K)}
=\frac{e^{-\eps R}}{\Gamma(d+1)\vol(\eps^{-1}K)}\mcom
\end{equation}
which is in agreement with~(\ref{eq:density}). That is, $g(a)=f(a).$
\end{remark}

\begin{figure}
\begin{center}
\fbox{
\begin{minipage}{.9\textwidth}
\vspace{1mm}
${\bf KM}(F,d,\eps)\colon$
\begin{enumerate}
\item Sample $z$ uniformly at random from $K=FB_1^n$ and sample
$r\sim\GammaDist(d+1,\eps^{-1}).$
\item Output $Fx+rz.$
\end{enumerate}
\vspace{1mm}
\end{minipage}
}
\end{center}
\label{fig:km}
\caption{Description of the $d$-dimensional $K$-norm mechanism.}
\end{figure}

The next theorem shows that the $K$-norm mechanism is indeed differentially
private. Moreover, we can express its error in terms of the \emph{expected
distance from the origin} of a random point in $K.$

\begin{theorem}\label{thm:knorm}
Let $\eps>0$. Suppose $F\colon\Rn\to\Rd$ is a linear map and put $K=FB_1^n.$
Then, the mechanism ${\rm KM}(F,d,\eps)$ is
$\eps$-differentially private, and for every
$p>0$ achieves the error bound $\E_{a\sim \mu_x} \|Fx-a\|^p \leq \frac{\Gamma(d+1+p)}{\eps^p\Gamma(d)}\E_{z\in
K}\|z\|_2^p.$ In particular, the $\ell_2$-error is at most
$\frac{d+1}{\eps}\E_{z\in K}\|z\|_2.$
\end{theorem}

\begin{proof}
To argue the error bound,
we will follow Remark~\ref{rem:gamma}. Let $D=\GammaDist(d+1,1/\eps).$
For all~$x\in\Rn$,
\begin{align*}
\E_{a\sim \mu_x}\|Fx-a\|^p
= \E_{a\sim \mu_0} \|a\|^p
= \E_{r\sim D} \E_{a\in rK} \|a\|^p
&= \left[\E_{r\sim D} r^p\right] \E_{z\in K} \|z\|^p\\
&= \frac{\Gamma(d+1+p)}{\eps^p\Gamma(d+1)}\E_{z\in K}\|z\|^p \tag{by
Fact~(\ref{fact:gamma})}.
\end{align*}
When $p=1$,
$\frac{\Gamma(d+1+p)}{\Gamma(d+1)}=d+1.$

Privacy follows from the fact that the mechanism is a special case of the
exponential mechanism~\cite{McSherryTa07}. For completeness, we repeat
the argument.

Suppose that $\|x\|_1\le1$. It suffices to show that for all $a\in\Rd$, the
densities of $\mu_0$ and $\mu_x$ are within multiplicative $\exp(\eps)$, i.e.,
\begin{align*}
\frac{Z^{-1}e^{-\eps\|a\|_K}}
{Z^{-1}e^{-\eps\|Fx-a\|_K}}
 = e^{\eps(\|Fx-a\|_K - \|a\|_K)}
 \le e^{\eps\|Fx\|_K}
 \le e^\eps.
\end{align*}
where in the first inequality we used the triangle inequality for
$\|\cdot\|_K$. In the second step we used
that $x\in B_1^n$ and hence $Fx\in FB_1^n=K$ which means $\|Fx\|_K\le1.$

Hence, the mechanism satisfies $\eps$-differential privacy.
\end{proof}

\section{Matching bounds for random queries}
\label{sec:random}

In this section, we will show that our upper bound matches our lower bound
when~$F$ is a random query.
A key observation is that $FB_1^n$ is the \emph{symmetric} convex hull
of $n$ (random) points $\{v_1,\dots,v_n\}\sse\Rd$, i.e., the convex hull
of $\{\pm v_1,\dots,\pm v_n\}$, where $v_i \in \Rd$ is the $i$th column of $F$.
The symmetric convex hull of random points has been
studied extensively in the theory of random polytopes.
A recent result of Litvak, Pajor, Rudelson and
Tomczak-Jaegermann~\cite{LitvakPaRuJa05} gives
the following lower bound on the volume of the convex hull.

\begin{theorem}[\cite{LitvakPaRuJa05}]
Let $2d\le n\le 2^d$ and let $F$ denote a random $d\times n$ Bernoulli matrix.
Then,
\begin{equation}\label{eq:vol}
\textstyle
\vol(FB_1^n)^{1/d}\ge \Omega(1)\sqrt{\log(n/d)/d}\mcom
\end{equation}
with probability
$1-\exp(-\Omega(d^\beta n^{1-\beta}))$ for any $\beta\in(0,\frac12).$
Furthermore, there is an explicit construction of $n$ points in $\{-1,1\}^d$
whose convex hull achieves the same volume.
\end{theorem}

We are mostly interested in the range where $n\gg d\log d$ in which case
the theorem was already proved by Giannopoulos and Hartzoulaki~\cite{GiannopoulosHa02}
(up to a weaker bound in the probability and without the explicit construction).

The bound in~(\ref{eq:vol}) is tight up to constant factors.
A well known result~\cite{BaranyFu88} shows that
the volume of the convex hull of any $n$ points on the sphere in $\Rd$
of radius $\sqrt{d}$ is bounded by
%
\begin{equation}
\vol(K)^{1/d}\le O(1)\sqrt{\log(n/d)/d}\mper
\end{equation}
%
Notice, that in our case
$K=FB_1^n\sse B_\infty^d\sse\sqrt{d}B_2^d$
and in fact
the vertices of $K$ are points on the $(d-1)$-dimensional sphere of radius
$\sqrt{d}$. However,
equation~(\ref{eq:vol}) states that the normalized volume of the random
polytope $K$ will be proportional to the volume of the Euclidean ball
of radius $\sqrt{\log(n/d)}$ rather than $\sqrt{d}.$
When $d\gg\log n$, this means that the volume of $K$ will be tiny compared to
the volume of the infinity ball~$B_\infty^d$.
%
%
%
By combining the volume lower bound with Theorem~\ref{thm:volume},
we get the following lower bound on the error of private mechanisms.
\begin{theorem}\label{thm:lower}
Let $\eps>0$ and $0<d\le n/2$. Then, for almost all
matrices $F\in\{-1,1\}^{d\times n}$, every
$\eps$-differentially private mechanism $M$
must have
\begin{equation}
\err(M,F)\ge\Omega(d/\eps)\cdot\min\left\{\sqrt{d},\sqrt{\log(n/d)}\right\}.
\end{equation}
\end{theorem}

\subsection{A separation result.}
We use this paragraph to point out that our lower bound immediately implies a
separation between approximate and exact differential privacy.

Theorem~\ref{thm:gaussian} gives a mechanism providing $\delta$-approximate
$\eps$-differential privacy 
with error $o(\eps^{-1}\sqrt{\log (n/d)})$ as long
as $\delta\ge 1/n^{o(1)}.$ Our lower bound  in Theorem~\ref{thm:lower}
on the other hand states that the error of any $\eps$-differentially private
mechanism must be $\Omega(\eps^{-1}\sqrt{\log(n/d)})$
(assuming $d\gg\log(n)$). We get the strongest separation when $d\le\log(n)$ and
$\delta$ is constant. In this case, our lower bound is a factor $\sqrt{d}$
larger than the upper bound for approximate differential privacy.

\subsection{Upper bound on average Euclidean norm}

Klartag and Kozma~\cite{KlartagKo09}
recently gave a bound on the quantity $\E_{z\sim K}\|z\|$ when
$K=FB_1^n$ for random~$F.$

\begin{theorem}[\cite{KlartagKo09}]
Let $F$ be a random $d\times n$ Bernoulli matrix and put $K=FB_1^n$. Then,
there is a constant $C>0$ so that with probability greater than $1-Ce^{-O(n)}$,
\begin{equation}
\frac1{\vol(K)}\int_{z \in K}\|z\|^2\rd z \le C\log(n/d).
\end{equation}
\end{theorem}

An application of Jensen's inequality thus gives us the following corollary.
\begin{corollary}
Let $\eps>0$ and $0<d\le n/2$. Then, for almost all
matrices $F\in\{-1,1\}^{d\times n}$, the mechanism ${\rm KM}(F,d,\eps)$ is
$\eps$-differentially private with error at most
\begin{equation}
 O(d/\eps)\cdot\min\left\{\sqrt{d},\sqrt{\log(n/d)}\right\}.
\end{equation}
\end{corollary}

\section{Approximately isotropic bodies}
\label{sec:isotropic}

The following definition is a relaxation of nearly isotropic position used
in literature (e.g.,~\cite{KannanLoSi97})
\begin{definition}[Approximately Isotropic Position]
We say a convex body $K \sse \Re^d$ is in $c$-\emph{approximately isotropic
position} if for every unit vector $v \in \Re^d$,
\begin{equation}
\frac1{\vol(K)}\int_K |\langle z,v\rangle|^2 \rd z
\le
c^2L_K^2 \vol(K)^{\frac{2}{d}} \mper
\end{equation}
\end{definition}

The results of Klartag and Kozma~\cite{KlartagKo09} referred to in the
previous section show that the symmetric convex hull $n$ random points from
the $d$-dimensional hypercube are in $O(1)$-approximately isotropic position and have
$L_K = O(1)$. More generally, the $K$-norm mechanism can be shown to be
approximately optimal whenever $K$ is nearly isotropic.
\begin{theorem}[Theorem~\ref{thm:upper1} restated]\label{thm:upper}
Let $\eps>0$. Suppose $F\colon\Rn\to\Rd$ is a linear map such that $K=FB_1^n$
is in $c$-approximately isotropic position. Then, the $K$-norm mechanism is
$\eps$-differentially private and has error at most
$O(c L_K)\cdot \vollb(F,\eps).$
\end{theorem}

\begin{proof}
By Theorem~\ref{thm:knorm}, the $K$-norm mechanism is $\eps$-differentially
private and has error
$\frac{d+1}\eps\E_{z\sim K}\|z\|.$
By the definition of the approximately isotropic position, we have:
$\E_{z\sim K} \|z\|^2
\leq d \cdot c^2 L_K^2 \vol(K)^{2/d}.$
%
By Jensen's inequality,
\[
\frac{d+1}\eps\E_{z \sim K}\|z\|
\le\frac{d+1}\eps \sqrt{\E_{z\sim K}\|z\|^2}
\le O(\eps^{-1}d\sqrt{d}\cdot\vol(K)^{1/d}cL_K).
\]
Plugging in the definition of $\vollb$ proves the result.
\end{proof}

We can see that the previous upper bound is tight up to a factor of~$cL_K$.
Estimating $L_K$ for general convex bodies is a well-known open problem in
convex geometry.
The best known upper bound for a general convex body
$K\subseteq \Rd$ is $L_K\le O(d^{1/4})$ due to Klartag~\cite{Klartag06},
improving over the estimate $L_K\le
O(d^{1/4}\log d)$ of Bourgain from '91.
The conjecture is that $L_K=O(1)$.

\begin{conjecture}[Hyperplane Conjecture]
\label{conj:hyperplane}
There exists $C>0$ such that for every $d$ and every convex set
$K\sse\Rd$, $L_K<C$.
\end{conjecture}

Assuming this conjecture we get matching bounds for approximately isotropic
convex bodies.
\begin{theorem}
Let $\eps>0.$ Assuming the hyperplane conjecture, for every
$F\in[-1,1]^{d\times n}$ such that $K=FB_1^n$ is $c$-approximately isotropic,
the $K$-norm mechanism ${\rm KM}(F,d,\eps)$ is  $\eps$-differentially private with
error at most
\begin{equation}
 O(c)\cdot \vollb(F,\eps) \leq
O(cd/\eps)\cdot\min\left\{\sqrt{d},\sqrt{\log(n/d)}\right\}\mper
\end{equation}
\end{theorem}

\section{Non-isotropic bodies}
\label{sec:nonisotropic}

While the mechanism of the previous sections is near-optimal for
near-isotropic queries, it can be far from optimal if $K$ is far from
isotropic. For example, suppose the matrix $F$ has random entries from
$\{+1,-1\}$ in the first row, and (say) from $\{\frac{1}{d^2},-\frac{1}{d^2}\}$ in the
remaining rows. While the Laplacian mechanism will add $O(\frac{1}{\epsilon})$
noise to the first co-ordinate of $Fx$, the $K$-norm mechanism will add noise
$O(d/\epsilon)$ to the first co-ordinate. Moreover, the volume lower
bound $\vollb$ is at most $O(\eps^{-1}\sqrt{d})$. Rotating $F$ by a random
rotation gives, w.h.p., a query for which the Laplacian mechanism adds
$\ell_2$ error $O(d/\eps)$. For such a body, the Laplacian and the $K$-norm
mechanisms, as well as the $\vollb$ are far from optimal.

In this section, we will design a recursive mechanism that can handle
such non-isotropic convex bodies. To this end, we will need to introduce a few more
notions from convex geometry.

Suppose $K\subseteq\Rd$ is a centered convex body, i.e. $\int_K
x\rd x=0.$
The \emph{covariance matrix of $K$}, denoted $M_K$ is the $d\times d$ matrix
with entry $ij$ equal to
$M_{ij} = \frac1{\vol(K)}\int_{K} x_ix_j\rd x.$ That is, $M_K$ is the
covariance matrix of the uniform distribution over $K.$
%

\subsection{A recursive mechanism}
Having defined the covariance matrix, we can
describe a recursive mechanism for the case when $K$ is not in
isotropic position. The idea of the mechanism is to act differently on
different eigenspaces of the covariance matrix. Specifically, the mechanism
will use a lower-dimensional version of ${\bf KM}(F,d',\eps)$ on subspaces
corresponding to few large eigenvalues.

Our mechanism, called ${\rm NIM}(F,d,\eps)$, is given a linear mapping
$F\colon\Rn\to\Rd,$ and parameters $d\in\mathbb{N},\eps>0.$
The mechanism proceeds recursively by partitioning the convex body $K$ into two
parts defined by the middle eigenvalue of $M_K.$ On one part it will act
according to the $K$-norm mechanism. On the other part, it will descend
recursively. The mechanism is described in Figure~\ref{fig:nim}

\begin{figure}
\label{fig:nim}
\begin{center}
\fbox{
\begin{minipage}{.9\textwidth}
\vspace{1mm}
${\bf NIM}(F,d,\eps)\colon$
\begin{enumerate}
\item Let $K=FB_1^n$. Let $\sigma_1\ge\sigma_2\ge\dots\ge\sigma_d$
denote the eigenvalues of the covariance matrix $M_K.$ Pick a corresponding
orthonormal eigenbasis $u_1,\dots,u_d$.
\item
\label{step:subspaces}
Let $d'=\lfloor d/2\rfloor$ and let $U={\rm span}\{u_1,\dots,u_{d'}\}$ and
$V={\rm span}\{u_{d'+1},\dots,v_d\}.$
\item \label{step:knorm}
Sample $a\sim{\bf KM}(F,d,\eps)\mper$
\item If $d=1$, output $P_V a$. Otherwise, output
${\bf NIM}(P_U F,d',\eps) + P_V a\mper$
\end{enumerate}
\vspace{1mm}
\end{minipage}
}
\end{center}
\caption{Mechanism for non-isotropic bodies}
\end{figure}

\begin{remark}
The image of $P_UF$ above is a $d'$-dimensional subspace of $\Rd.$ We assume
that in the recursive call ${\rm NIM}(P_UF,d',\eps)$, the $K$-norm mechanism
is applied to a basis of this subspace. However, formally the output is a
$d$-dimensional vector.
\end{remark}

%

To analyze our mechanism, first observe that the recursive calls
terminate after at most $\log d$ steps.
For each recursive step $m\in\{0,\dots,\log d\}$, let $a_m$ denote the
distribution over the output of the $K_m$-norm mechanism in step 3. Here, $K_m$
denotes the $d_m$-dimensional body given in step $m.$

\begin{lemma}\label{lem:nim-privacy}
The mechanism ${\rm NIM}(F,d,\eps)$ satisfies $(\eps \log d)$-differential privacy.
\end{lemma}

\begin{proof}
We claim that for every step $m\in\{0,\dots,\log d\}$, the distribution over
$a_m$ is $\eps$-differentially private. Notice that this claim implies the
lemma, since the joint distribution of $a_0,a_1,\dots,a_m$ is
$\eps\log(d)$-differentially private. In particular, this is true for the
final output of the mechanism as it is a function of $a_0,\dots,a_m.$

To see why the claim is true, observe that each $K_m$ is the $d_m$-dimensional
image of the $\ell_1$-ball under a linear mapping. Hence, the $K_m$-norm
mechanism guarantees $\eps$-differential privacy by Theorem~\ref{thm:knorm}.
\end{proof}

The error analysis of our mechanism requires more work. In particular, we need
to understand how the volume of~$P_U K$ compares to the norm of~$P_Va.$ As a
first step we will analyze the volume of~$P_U K.$

\subsection{Volume in eigenspaces of the covariance matrix}
\label{sec:eigenvol}
Our goal in this section is to express the volume of $K$ in eigenspaces of the
covariance matrix in terms of the eigenvalues of the covariance matrix. This
will be needed in the analysis of our mechanism for non-isotropic bodies.

We start with a formula for the volume of central sections of
isotropic bodies. This result can be found in \cite{MilmanPa89}.

\begin{proposition}\label{prop:volumeisoproj}
Let $K\subseteq\Rd$ be an isotropic body of unit volume. Let $E$ denote a
$k$-dimensional subspace for $1\le k\le d$. Then,
\[
\vol_{k}(E\cap K)^{1/(d-k)} = \Theta\left(\frac{L_{B_K}}{L_K}\right).
\]
Here, $B_K$ is an explicitly defined isotropic convex body.
\end{proposition}

From here on, for an isotropic body $K$, let $\alpha_K = \Omega (L_{B_K}/L_K)$
be a lower bound on $\vol_{k}(E\cap K)^{1/(d-k)}$ implied by the above
proposition. For a non-isotropic $K$, let $\alpha_K$ be $\alpha_{TK}$ when $T$
is the map the brings $K$ into isotropic position. Notice that if the
Hyperplane Conjecture is true, then $\alpha_K=\Omega(1).$ Moreover, $\alpha_K$
is $\Omega(d^{\frac{1}{4}})$ due to the results of~\cite{Klartag06}.

\begin{corollary}
\label{cor:slicevolume}
Let $K\subseteq\Rd$ be an isotropic
body with $\vol(K)=1.$ Let $E$ denote a $k$-dimensional subspace for $1\le k\le
d$ and let $P$ denote an orthogonal projection operator onto the
subspace $E$. Then,
\[
\vol_{k}(P K)^{1/(d-k)} \geq \alpha_K.
\]
\end{corollary}
\begin{proof}
Observe that the $PK$ contains $E \cap K$ since $P$ is the identity on $E.$
\end{proof}
%

We cannot immediately use these results since they only apply to isotropic
bodies and we are specifically dealing with non-isotropic bodies. The trick is
to apply the previous results after transforming $K$ into an isotropic body
while keeping track how much this transformation changed the volume.

As a first step, the following lemma relates the volume of projections of an
arbitrary convex body $K$ to the volume of projections of $TK$ for some linear
mapping~$T$.

\begin{lemma}\label{lem:volumeproj}
Let $K\subseteq\Rd$ be a symmetric convex body. Let $T$ be a linear map which
has eigenvectors $u_1,\ldots,u_d$ with eigenvalues
$\lambda_1,\ldots,\lambda_d$. Let $1\le k\le d$ and suppose
$E={\rm span}\{u_1,u_2,\dots,u_k\},$ Denote by $P$ be the projection operator
onto the subspace $E.$ Then,
\[
\vol_{k}(PK)
\ge \vol_{k} (PTK)
\prod_{i=1}^k\lambda_i^{-1} \mper
\]
\end{lemma}

\begin{proof}
For simplicity, we assume that the eigenvectors of $T$ are the standard basis vectors $e_1,\ldots,e_d$; this is easily achieved by applying a rotation to $K$.
Now, it is easy to verify that
$P = PT^{-1} T  = S PT$
where
$S=\rm{diag}(\lambda_1^{-1},\lambda_2^{-1},\ldots, \lambda_k^{-1},0,\dots,0)$.
Thus we can write
\begin{displaymath}
\vol_k(PK)
= \det(S_{|E}) \vol_k (PTK) =\frac1{\prod_{i=1}^k \lambda_i} \vol_k(P
TK)\mper\qedhere
\end{displaymath}
\end{proof}
Before we can finish our discussion,
we will need the fact that the isotropic constant of $K$ can be
expressed in terms of the determinant of $M_K.$
\begin{fact}[\cite{Giannopoulos03,MilmanPa89}]
\label{fact:isodet}
Let $K\sse\Rd$ be a convex body of unit volume. Then,
\begin{equation}
L_K^2\vol(K)^{\frac{2}{d}} = \det(M_K)^{1/d}.
\end{equation}
Moreover, $K$ is in isotropic position iff $M_K =
L_K^2\vol(K)^{2/d}I$.
\end{fact}

We conclude with the following~Proposition~\ref{prop:nivol}.
\begin{proposition}
\label{prop:nivol}
Let $K\subseteq\Rd$ be a symmetric convex body. Let $M_k$ have eigenvectors
$u_1,\ldots,u_d$ with eigenvalues $\sigma_1,\ldots,\sigma_d$.
Let $1\le k\le \lceil\frac{d}{2}\rceil$ with  and suppose
$E={\rm span}\{u_1,u_2,\dots,u_k\},$ Denote by $P$ be the
projection operator onto the subspace $E.$ Then,
\begin{equation}
\vol_{k}(PK)^{1/(d-k)}
\ge \Omega(1)\cdot\alpha_K
\left(\prod_{i=1}^k\sigma_i^{1/2}\right)^{1/(d-k)},
\end{equation}
where $\alpha_K$ is $\Omega(1/d^{\frac{1}{4}})$. Moreover,
assuming the Hyperplane conjecture, $\alpha_K\ge\Omega(1)$.
\end{proposition}

\begin{proof}
Consider the linear mapping $T=M_K^{-\half}\mper$ this is well defined since
$M_K$ is a positive symmetric matrix. It is easy to see that after applying
$T$, we have $M_{TK}=I.$ Hence, by Fact~\ref{fact:isodet}, $TK$ is in
isotropic position and has volume $\vol(TK)^{1/d}=1/L_{TK}=1/L_K,$ since
$\det(M_{TK})=1.$ Scaling $TK$ by $\lambda=L_K^{1/d}$
hence results in $\vol(\lambda TK)=1.$
Noting that $\lambda T$ has eigenvalues
$\lambda\sigma_1^{-\frac12},\lambda\sigma_2^{-\frac12},\ldots,
\lambda\sigma_d^{-\frac12}$,
we can apply Lemma~\ref{lem:volumeproj} and get
\[
\vol_{k}(PK)
\ge \vol_k(P\lambda TK)\prod_{i=1}^k\frac{\sqrt{\sigma_i}}{\lambda}
\]

Since $\lambda TK$ is in isotropic position and has unit volume,
Corollary~\ref{cor:slicevolume} implies that
\begin{equation}
\label{eq:isoprojvol}
\vol_{k}(P\lambda TK)^{1/(d-k)} \geq \alpha_K\mper
\end{equation}
Thus the required inequality holds with an additional $\lambda^{-\frac{k}{d-k}}$ term. By assumption on $k$, $\frac{k}{d-k}$ is at most $2$. Moreover, $\lambda = L_K^{1/d} \leq d^{1/d} \leq 2$, so that this additional term is a constant. As discussed above, $\alpha_K$ is $\Omega(d^{-\frac{1}{4}})$ by~\cite{Klartag06}, and $\Omega(1)$ assuming the Hyperplane Conjecture~\ref{conj:hyperplane}. Hence the claim.
\end{proof}

\subsection{Arguing near optimality of our mechanism}
\label{sec:nim}

Our next lemma shows that the expected squared Euclidean error
added by our algorithm in each step is bounded by the square of the optimum.
We will first need the following fact.
\begin{fact}\label{fact:eigenvalues}
Let $K\subseteq\Rd$ be a centered convex body.
Let $\sigma_1\ge\sigma_2\ge\dots\ge\sigma_d$
denote the eigenvalues of $M_K$ with a corresponding
orthonormal eigenbasis $u_1,\dots,u_d.$ Then, for all $1\le i\le d$,
\begin{equation}
\sigma_i
=\max_{\theta} \E_{x\in K} \langle \theta,x \rangle^2
\end{equation}
where the maximum runs over all $\theta\in\mathbb{S}^{d-1}$ such that $\theta$
is orthogonal to $u_1,u_2,\dots,u_{i-1}.$
\end{fact}

\begin{lemma}
Let $a$ denote the random variable returned by the $K$-norm mechanism in
step~(\ref{step:knorm}) in the above description of ${\rm NIM}(F,d,\eps)$. Then,
\[
\Gvollb(F,\eps)^2
\ge \Omega(\alpha_K^2)\E\|P_Va\|_2^2\mper
\]
\end{lemma}

\begin{proof}
For simplicity, we will assume
that $d$ is even and hence $d-d'=d'.$
The analysis of the $K$-norm mechanism (Theorem~\ref{thm:knorm} with $p=2$)
shows that the random variable $a$ returned by the $K$-norm mechanism in
step~(\ref{step:knorm}) satisfies
\begin{align}
\E\|P_Va\|_2^2
 = \frac{\Gamma(d+3)}{\eps^2\Gamma(d+1)}
& = \frac{(d+2)(d+1)}{\eps^2}\E_{z\in K}\|P_V z\|_2^2 \notag \\
& = O\left(\frac {d^2}{\eps^2}\right)
\sum_{i=d'+1}^d \E_{z\in K} \langle z,u_i\rangle^2 \notag \\
& = O\left(\frac {d^2}{\eps^2}\right)
\sum_{i=d'+1}^d {\sigma_i}\tag{by
Fact~\ref{fact:eigenvalues}}\\
&\le O\left(\frac{d^3}{\eps^2}\right)\cdot \sigma_{d'+1}.
\label{eq:errorupper}
\end{align}
On the other hand, by the definition of $\Gvollb$,
\begin{align*}
\Gvollb(F,\eps)^2
& \ge \Omega\left(\frac{d^3}{\eps^2}\right)
\cdot\vol_{d'}(P_UK)^{2/d'} \\
& \ge \Omega\left(\frac{d^3}{\eps^2}\right)
\alpha_K^2 \left(\prod_{i=1}^{d'}\sigma_i\right)^{1/d'}
\tag{by Proposition~\ref{prop:nivol}}\\
& \ge \Omega\left(\frac{d^3}{\eps^2}\right)
\alpha_K^2\sigma_{d'}.
\end{align*}
Since $\sigma_{d'} \ge \sigma_{d'+1}$, it follows that
\begin{equation}
\Gvollb(F,\eps)^2
\ge \Omega(\alpha_K^2)\E\|P_Va\|^2\mper
\end{equation}
The case of odd $d$ is similar except that we define $K'$ to be the projection onto the first $d'+1$ eigenvectors.
\end{proof}

\begin{lemma}\label{lem:nim-error}
Assume the hyperplane conjecture. Then,  the $\ell_2$-error of the
mechanism ${\rm NIM}(F,d,\eps)$ satisfies
\[
\err({\rm NIM},F)\le O(\sqrt{\log(d)}\cdot\Gvollb(F,\eps)).
\]
\end{lemma}

\begin{proof}
We have to sum up the error over all recursive calls of the mechanism. To
this end, let $P_{V_m}a_m$ denote the output of the $K$-norm
mechanism $a_m$ in step $m$ projected to the corresponding subspace $V_m$.
Also, let $a\in\Rd$ denote the final output of our
mechanism. We then have,
\begin{align*}
\E\|a\|_2
&
\le \sqrt{\E\|a\|_2^2}
\tag{Jensen's inequality}\\
& =\sqrt{\sum_{m=1}^{\log d}\E\|P_{V_m}a_m\|_2^2} \\
& \le\sqrt{\sum_{m=1}^{\log d} O(\alpha_{K_m}^{-2})\cdot\Gvollb(F,\eps)^2}
\tag{by Lemma~\ref{lem:nim-error}}\\
& \le O(\sqrt{\log d})\left(\max_m\alpha_{K_m}^{-1}\right)\Gvollb(F,\eps).
\end{align*}
Here we have used the fact that $\Gvollb(F,\eps) \geq \Gvollb(P_UF,\eps)$.
Finally, the hyperplane conjecture implies
$\max_m \alpha_{K_m}^{-1} = O(1).$
\end{proof}

\begin{corollary}
Let $\eps>0$. Suppose $F\colon\Rn\to\Rd$ is a linear map. Further,
assume the hyperplane conjecture.  Then, there is an
$\eps$-differentially private mechanism~$M$ with error
\[
\err(M,F)\le O(\log(d)^{3/2}\cdot \Gvollb(F,\eps)).
\]
\end{corollary}

\begin{proof}
The mechanism ${\rm NIM}(F,d,\eps/\log(d))$ satisfies $\eps$-differential
privacy, by Lemma~\ref{lem:nim-privacy}. The error is at most
$\log(d)\sqrt{\log d}\cdot \Gvollb(F,\eps)$ as a direct consequence of
Lemma~\ref{lem:nim-error}.
\end{proof}

Thus our lower bound $\Gvollb$ and the mechanism ${\rm NIM}$ are both within $O(\log^{3/2} d)$ of the optimum.

\section{Efficient implementation of our mechanism}
\label{sec:efficient}

We will first describe how to implement our mechanism in the case where $K$ is
isotropic. Recall that we first sample $R\sim\GammaDist(d,\eps^{-1})$ and then
sample a point $a$ uniformly at random from $R K.$ The first step poses no
difficulty. Indeed when $U_1,\dots,U_d$ are independently distributed
uniformly over the interval $(0,1]$, then a standard fact tells us that
$\eps^{-1}\sum_{i=1}^d -\ln(U_i) \sim \GammaDist(d,\eps^{-1}).$
Sampling uniformly from $K$ on the other hand may be hard.
However, there are ways of sampling
nearly uniform points from $K$ using various types of rapidly mixing random
walks. In this section, we will use the \emph{Grid Walk} for
simplicity even though there are more efficient walks that will work for us. 
We refer the reader to the survey of
Vempala~\cite{Vempala05} or the original paper of Dyer, Frieze and
Kannan~\cite{DyerFrKa91} for a description of the Grid walk and
background information. Informally, the Grid walk samples nearly uniformly
from a grid inside $K$, i.e., ${\cal L}\cap K$ where we
take ${\cal L}=\frac1{d^2} \mathbb{Z}^d.$
The Grid Walk poses two requirements on $K$:
\begin{enumerate}
\item Membership in $K$ can be decided efficiently.
\item $K$ is bounded, in the sense that $B_2^d\sse K \sse d B_2^d$.
\end{enumerate}
Both conditions are naturally satisfied in our case where $K=FB_1^n$ for some
$F\in[-1,1]^{d\times n}$. Indeed, $K\sse B_\infty^d\sse\sqrt{d}B_2^d$ and we may
always assume that $B_2^d\sse K$, for instance, by considering
$K'=K + B_2^d$ rather than $K$. This will only increase the noise level
by $1$ in Euclidean distance. Notice that $K'$ is convex.
In order to implement the membership oracle for $K$, we need to be able to
decide for a given $a\in\Rd$, whether there exists an $x\in B_1^n$ such
that $Fx=a$. These constraints can be encoded using a linear program. In the
case of $K'$ this can be done using convex programming~\cite{GrotschelLS94}.

The mixing time of the Grid walk is usually quantified in terms of the total
variation (or $L_1$) distance between the random walk and its stationary
distribution. The stationary distribution of the grid Walk is the uniform
distribution over ${\cal L}\cap K$.
Standard arguments show that an $L_1$-bound gives us $\delta$-approximate
$\eps$-differential privacy where $\delta$ can be made exponentially small in
polynomial time.
In order to get exact privacy ($\delta=0$) we instead need a multiplicative
guarantee on the density of the random walk at each point in $K.$

In Appendix~\ref{sec:mixing}, we show that the Grid Walk actually
satisfies mixing bounds in the relative $L_\infty$-metric which gives us the
following theorem. We also need to take care of the fact that the stationary
distribution is a priori not uniform over $K.$ A solution to this problem is
shown in the appendix as well.
\begin{theorem}
\label{thm:ballwalk}
Let $P_t$ denote the Grid Walk over $K$ at time step $t$. Given a linear
mapping $F\colon\Rn \to\Rd$ and $x\in\Rn$, consider the mechanism $M'$ which
samples $R\sim\GammaDist(d+1,\eps^{-1})$ and then outputs $Ra$ where $a\sim P_t.$
Then, there is some $t\le\poly(d,\eps^{-1})$ such that
\begin{enumerate}
\item $M'$ is $O(\eps)$-differentially private,
\item $\err(M',F)=\err(M,F)+O(1)$, where $M$ denotes the $K$-norm mechanism.
\end{enumerate}
\end{theorem}
We conclude that the Grid walk gives us an efficient implementation of our
mechanism which achieves the same error bound (up to constants) and
$\eps$-differential privacy.
\begin{remark}
The runtime stated in Theorem~\ref{thm:ballwalk} depends only upon $d$ and
$\eps^{-1}.$ The polynomial dependence on $n$ only comes in when implementing
the separation oracle for $K$ as described earlier. Since we think of $d$ as
small compared to $n$, the exact runtime of our algorithm heavily depends
upon how efficiently we can implement the separation oracle.
\end{remark}

\subsection{When $K$ is not isotropic}

In the non-isotropic case we additionally need to compute the subspaces $U$
and $V$ to project onto (Step~\ref{step:subspaces} of the algorithm).
Note that these subspaces themselves depend only on
the query $F$ and not on the database $x$. Thus these can be published and the
mechanism maintains its privacy for an arbitrary choice of subspaces $U$ and
$V$. The choice of $U,V$ in Section~\ref{sec:nonisotropic} depended on the
covariance matrix $M$, which we do not know how to compute exactly. We next
describe a method to choose $U$ and $V$ that is efficient such that the
resulting mechanism has essentially the same error. The sampling from $K$ can
then be replaced by approximate sampling as in the previous subsection,
resulting in a polynomial-time differentially private mechanism
with small error.

Without loss of generality, $K$ has the property that $B_2^d\sse K \sse d^2
B_2^d$. In this case, $x_ix_j \leq d^4$ so that with $O(d^4 \log d)$
(approximately uniform) samples from $K$, Chernoff bounds imply that the
sample covariance matrix approximates the covariance matrix well. In other
words, we can construct a matrix $\tilde M_K$ such that each entry of $\tilde
M_K$ is within $neg(d)$ of the corresponding entry in $M_K$. Here and in the
rest of the section, ${\it neg}(d)$ denotes an negligible function bounded above by
$\frac{1}{d^C}$ for a large enough constant $C$, where the constant may vary
from one use to the next. Let the eigenvalues of $\tilde M$ be $\tilde
\sigma_1,\ldots,\tilde \sigma_d$ with corresponding eigenvectors $\tilde
u_1,\ldots,\tilde u_d$. Let $\tilde T$ be the $\tilde M_K^{-\frac{1}{2}}$, and
let $\tilde P$ be the projection operator onto the span of the first $d'$
eigenvectors of $\tilde M_K$. This defines our subspaces $\tilde U$ and
$\tilde V$, and hence the mechanism. We next argue that
Lemma~\ref{lem:nim-error} continues to hold.

First note that for any $i \geq d'+1$
\begin{align*}
\E_{a \in K} \langle a, \tilde u_i\rangle^2
&= \|\tilde u_i^T M_K \tilde u_i\| \\
&= \|\tilde u_i^T \tilde M_K \tilde u_i\|
  + \|\tilde u_i^T (M_K - \tilde M_K) \tilde u_i\|\\
&= \tilde \sigma_i + {\it neg}(d).
\end{align*}
Thus, Equation~\ref{eq:errorupper} continues to hold with $\tilde
\sigma_{d'+1}$ replacing $\sigma_{d'+1}$.

To prove that Proposition~\ref{prop:nivol} continues to hold (with $\tilde
M,\tilde T,\tilde P$ replacing $M,T,P$), we note that the only place in the
proof that we used that $M_K$ is in fact the covariance matrix of $K$ is
(\ref{eq:isoprojvol}), when we require $TK$ to be isotropic. We next argue
that (\ref{eq:isoprojvol}) holds for $\tilde{T}K$ if $\tilde{M_K}$ is a good
enough approximation to $M_K$. This would imply Proposition~\ref{prop:nivol} and
hence the result.

First recall that Wedin's theorem~\cite{Wedin72} states that for non-singular
matrices $R$, $\tilde R$,
\[
\|R^{-1}-\tilde R^{-1}\|_2 \leq
\frac{1+\sqrt{5}}{2}
\|R-\tilde R\|_2 \cdot \max\{\|R^{-1}\|_2^2,\|\tilde R^{-1}\|_2^2\}\mper
\]
Using this for the matrices $M^{\frac{1}{2}},\tilde M^{\frac{1}{2}}$ and using
standard perturbation bounds gives (see e.g.~\cite{KempeM08}):
\begin{equation}
\|\tilde T - T\|_2 \leq O(1)
\|T\|_2^2 \cdot \|\tilde M_K^{\frac{1}{2}} - M_K^{\frac{1}{2}}\|_2\mper
\end{equation}
Since
$\|T\|_2$ is at most $poly(d)$ and the second term is $neg(d)$, we conclude
that $\|\tilde T - T\|_2$ is $neg(d)$. It follows that $TK \subseteq \tilde T
K + neg(d) B_2^d$. Moreover, since $TK$ is in isotropic position, it contains
a ball $\frac{1}{d}B_2^d$. It follows from Lemma~\ref{lem:ballcontainment} in
the appendix that $\frac{1}{2d} B_2^d$ is contained in $\tilde T K$. Thus,
\begin{align*}
\left(1-\tfrac{1}{d}\right) TK
&\subseteq \left(1-\tfrac{1}{d}\right) \tilde T K + neg(d) B_2^d\\
&\subseteq \left(1-\tfrac{1}{d}\right) \tilde T K + neg(d) \tilde T K\\
&\subseteq \tilde T K\mcom
\end{align*}
where the last containment follows from
the fact that $\tilde T K$ is convex and contains the origin. Thus
$(1-\frac{1}{d})\tilde P TK \subseteq \tilde P \tilde T K$. Since
Corollary~\ref{cor:volumesubspace} still lower bounds the volume of $\tilde  P
TK$, we conclude that
\[
\vol_{k}(\tilde P \tilde TK)^{1/k}
\geq \frac{1}{e}\vol_{k}(\tilde P TK)^{1/k}
\geq \frac{\alpha_K^{\frac{d-k}{k}}}{e}\mcom
\]
where we have used the fact that $k
\leq d$ so that $(1-\frac{1}{d})^k \geq \frac{1}{e}$. For $k=d'$,
$\frac{d-k}{k}$ is $\Theta(1)$ so that $\vol_{k}(\tilde P \tilde T
K)^{1/(d-k)} \geq \Omega(\alpha_K)$.  Thus we have shown that up to
constants, (\ref{eq:isoprojvol}) holds for $\vol_{k}(\tilde P \tilde T K)^{1/(d-k)}$
which completes the proof.

\section{Generalizations of our mechanism}
\label{sec:generalize}

Previously, we studied linear mappings $F\colon\Rn\to\Rd$ where $\Rn$ was
endowed with the $\ell_1$-metric. However, the $K$-norm mechanism is
well-defined in a much more general context. The only property of $K$
used here is its convexity. In general, let $\dee$ be an arbitrary domain of
databases with a distance function $\dist$. Given a function
$F : \dee \rightarrow \Re^d$, we could define
$K_0 = \{(F(x) - F(x'))/{\it dist}(x,x'): x,x' \in \dee\}$
and let $K$ be the convex closure of $K_0$. Then the $K$-norm mechanism can be
seen to be differentially private with respect to ${\it dist}$. Indeed note
that that $|q(d,a)-q(d',a)| = |F(d)-a|_K - |F(d')-a|_K \leq |F(d)-F(d')|_K
\leq {\it dist}(d,d')$, and thus privacy follows from the exponential
mechanism.

Moreover, in cases when one does not have a good handle on $K$ itself, one can
use any convex body $K'$ containing $K$.

\paragraph{Databases close in $\ell_2$-norm.}
For example,
McSherry and Mironov~\cite{McSherryM09} can transform their input data set so
that neighboring databases map to points within Euclidean distance at most $R$
for a suitable parameter $R$. Thus ${\it dist}$ here is the $\ell_2$ norm and for
any linear query, $K$ is an ellipsoid. 
%
\paragraph{Local Sensitivity.}
Nissim, Raskhodnikova and Smith~\cite{NissimRS07} define {\em smooth
sensitivity} and show that one can design approximately differentially private
mechanism that add noise proportional to the smooth sensitivity of the query.
This can be significant improvement when the local sensitivity is much smaller
than the global sensitivity. Notice that such queries are necessarily
non-linear. We point out that one can define a local sensitivity analogue of
the $K$-norm mechanism by considering the polytopes
$K_x = {\rm conv}\left\{\frac{F(x')-F(x)}{\dist(x,x')}\colon
x'\in\dee\right\}$
and adapting the techniques of \cite{NissimRS07} accordingly.

\bibliographystyle{alpha}
\bibliography{noise}

\appendix

\section{Mixing times of the Grid Walk in $L_\infty$}
\label{sec:mixing}

In this section, we sketch the proof of Theorem~\ref{thm:ballwalk}.
We will be interested in the mixing properties of Markov
chains over some measured state space $\Omega.$ We
will need to compare probability measures $\mu,\nu$ over the space~$\Omega.$

The relative $L_\infty$-distance is defined as
\begin{equation}
\|\mu/\nu-1\|_{\infty}=\sup_{u\in
\Omega}\left|\frac{d\mu(u)}{d\nu(u)}-1\right|\mper
\end{equation}

For a Markov chain $P$,
we will be interested in the mixing time in the $\infty$-metric.
That is the smallest number $t$ such that $\|P_t/\pi-1\|_\infty\le\eps.$
Here, $P_t$ is the distribution of $P$ at step $t$ and $\pi$ denotes
the stationary distribution of $P.$
The relevance of the $\infty$-norm for our purposes is given by the following
fact.
\begin{lemma}
\label{lem:infeps}
Suppose $M=\{\mu_x\}_{x\in\Rn}$ is an $\eps$-differentially private mechanism
$M$ and suppose $M'=\{\mu_x'\}_{x\in\Rn}$ satisfies
$\max\{\|\mu_x/\mu_x'-1\|_{\infty},\|\mu_x'/\mu_x-1\|_\infty\}\le\eps$
for some $0\le\eps\le1$ and all $x\in\Rn$. Then, $M'$ is
$3\eps$-differentially private.
\end{lemma}

\begin{proof}
By our second assumption,
\[
\max\left\{\frac{\rd\mu_x(u)}{\rd\mu_x'(u)},
\frac{\rd\mu_x'(u)}{\rd\mu_x(u)}\right\}\le
1+\eps\le e^\eps.
\]
where we used that $1+\eps\le e^\eps$ for $0\le\eps\le1.$

Now, let $x,x'$ satisfy $\|x-x'\|_1\le1$. By the previous inequality, we have
\[
\sup_{u\in\Omega} \frac{\rd \mu_x'(u)}{\rd \mu_{x'}'(u)}
\le\sup_{u\in\Omega} \frac{\rd \mu_x(u)e^\eps}{\rd \mu_{x'}(u)e^{-\eps}}
\le
e^{2\eps}
\sup_{u\in\Omega} \frac{\rd \mu_x(u)}{\rd \mu_{x'}(u)}
\le e^{3\eps}.
\]
In the last inequality, we used the assumption that $M$ is
$\eps$-differentially private. Hence, we have shown that $M'$ is
$3\eps$-differentially private.
\end{proof}

Now consider the grid walk with a fine enough grid (say side length $\beta$).
It is known that a random walk on a grid gets within statistical distance at
most $\Delta$ of the uniform distribution in time that is polynomial in
$d,\beta^{-1}$ and $\log \Delta^{-1}$. Setting $\Delta$ to be smaller than the
$\eps(\beta/d)^d$, we end up with a distribution that is within
$\ell_{\infty}$ distance at most $\eps$ from the uniform distribution on the
grid points in $K$. Let $\hat{z}$ be a sample from the grid walk, and let $z$
be a random point from an $\ell_{\infty}$ ball of radius half the side length
of the grid, centered at $\hat{z}$. Then $z$ is a (nearly) uniform sample from
a body $\tilde{K}$ which has the property that $(1-\beta)K \sse \tilde{K} \sse
(1+\beta)K$.

\subsection{Weak separation oracle}
An \emph{$\eta$-weak} separation oracle for $K'$ is a
blackbox that says `YES' when given $u\in\Rd$ with $(u+\eta B_2^d)\sse K'$ and
outputs `NO' when $u\not\in K'+\eta B_2^d.$ Here, $\eta>0$ is some parameter
that we can typically make arbitrarily small, with the running time depending
on $\eta^{-1}$. Our previous discussion assumed an oracle for which $\eta=0.$
Taking $\eta =\beta$ ensures that the sample above is (nearly) uniform from a
body $\hat{K}$ such that $(1-2\beta)K \sse \tilde{K} \sse (1+2\beta)K$. By
rescaling, we get the following lemma.

\begin{lemma}
Let $K$ be a convex body such that $B_2^d \sse K \sse dB_2^d$,
and let $\beta> 0$. Suppose $K$ is represented by a $\beta$-weak separation
oracle. Then, there is a randomized algorithm $Sample(K,\beta)$
running in time $poly(d,\beta^{-1})$ whose output distribution is within
$\ell_{\infty}$-distance at most $\beta$ from the uniform distribution over a
body $\hat{K}$ such that $K\sse \hat{K} \sse (1+\beta)K$.
\end{lemma}

We now argue that such a (nearly) uniform sample from a body close enough to
$K$ suffices for the privacy guarantee. Our algorithm first samples
$r\sim\GammaDist(d+1,\eps^{-1})$, and then outputs $Fx+rz$ where $z$ is the
output of $Sample(K,\beta)$ for $\beta = \min(\eps/d,1/r)$.

We can repeat the calculation for the density at a point $a$ in
equation~(\ref{eqn:densitycal}). Indeed for a point $a$ with $\|a-Fx\|_K=R$, the
density at $a$ conditioned on a sample $r$ from the Gamma distribution, is $
(1\pm \beta)/\vol(r\hat{K})$ whenever $(a/r) \in \hat{K}$, and zero otherwise.
By our choice of $\beta$, $\vol(\hat{K}) = (1\pm \eps)\vol(K)$. Moreover $(a/r)
\in \hat{K}$ for $r \geq R$ and $(a/r) \not\in \hat{K}$ for $r < R/(1+\beta)$.
Thus the density at $a$ is
\[ g(a) \geq \frac{1\pm
(\eps+\beta)}{\eps^{-d}\Gamma(d+1)} \int_{R}^\infty \frac{e^{-\eps
t}t^d}{\vol(tK)}\rd t =\frac{(1\pm (\eps+\beta))\int_{R}^\infty e^{-\eps t}\rd
t}{\Gamma(d+1)\vol(\eps^{-1}K)} =\frac{(1\pm (\eps+\beta))e^{-\eps
{R}}}{\Gamma(d+1)\vol(\eps^{-1}K)}.
\]
Similarly, $(a/r) \not\in \hat{K}$ for
$r < R/(1+\beta)$ implies that $g(a) \leq \frac{(1\pm (\eps+\beta))e^{-\eps
{R/(1+\beta)}}}{\Gamma(d+1)\vol(\eps^{-1}K)}$. It follows that $g(a)$ is
within an $\exp(O(\eps))$ factor of the ideal density.

Finally, the bound on the moments of the Gamma distribution from
Fact~\ref{fact:gamma} implies that the expected running time of this algorithm
is polynomial in $d,\eps^{-1}$.

\section{Lower bounds for Differential Privacy with respect to Hamming Distance}
\label{sec:lbhamming}
While our lower bounds were proved for differential privacy in the
$\ell_1$-metric, the usual notion of differential privacy
uses Hamming distance instead. In this
section we argue that for small enough~$\eps$, our lower bounds can be
extended to the usual definition. Let the database be a vector $w \in [n]^N$
where each individual has a private value in $[n]$. Such a database can be
transformed to its histogram $x= x(w) \in \Z_+^n$ where $x_i(w)$ denotes the
number of inputs that take value $i$, i.e. $x_i(w) = |\{j: w_j = i\}|$. A
linear query $F$ on the histogram is a sensitivity $1$ query on the database
$w$, and a mechanism $M$ is $\eps$-differentially private with respect to the
Hamming distance on $w$, if and only if it is differentially private with
respect to the $\ell_1$ norm, when restricted to non-negative integer vectors
$x$.

We can then repeat the proof of theorem~\ref{thm:volume}, with minor modifications to handle the non-negative integer constraint.

\begin{theorem}
Let $\eps>0$ and suppose $F\in \{-1,1\}^{d\times n}$ is a linear map and let $K=FB_1^n$.
Then, every $\eps$-differentially private mechanism $M$ for computing $G(w)=Fx(w)$
must have
\begin{equation}
\err(M,G) \ge \Omega(\vollb(F,\eps)),
\end{equation}
whenever $\eps < cd\vol(K)^{1/d}/\sqrt{n}$, for a universal constant $c$.
\end{theorem}

\begin{proof}
Let $R=\vol(K)^{1/d}$.
By Fact~\ref{fact:balls} and our assumption, $(d/4\eps) K = F ((d/4\eps)B_1^n)$
contains an $CRd\sqrt{d}/4\eps$-packing $Y\sse \Re^d$ of size at least $\exp(d)$, for some constant $C$. Let $X\subseteq (d/4\eps)B_1^n$ be a set of arbitrarily chosen preimages of $y \in Y$ so that $|X|=|Y|$ and $FX = Y$.

Now we come up with a similar set $X' \in Z_+^n$. For each $x \in X$, we round each $x_i$ randomly up or down, i.e. $\hat{x}_i = \lceil x_i \rceil$, with probability $(x_i-\lfloor x_i \rfloor)$, and $\lfloor x_i \rfloor$ otherwise. It is easy to check that $E[\hat{x}] = x$. so that with probability $2/3$, $|\hat{x}|_1 \leq 3 |x|_1$. Moreover, $E[F\hat{x}]=Fx$ and each random choice can change $\|F\hat{x}\|$ by at most $\sqrt{d}$. Thus martingale concentration results imply that with probability $2/3$, $\|F\hat{x}-Fx\| \leq  2\sqrt{dn}$. Thus there exists a choice of $\hat{x}$ so that both these events happen. Let $v$ denote the vector $(\lceil d/2\eps\rceil,\lceil d/2\eps\rceil,\ldots,\lceil d/2\eps\rceil)$ and set $x' = \hat{x} + v$. This defines our set $X'$ which is easily seen to be in $Z_+^n$. In fact, $X' \sse v + (\lceil d/2\eps\rceil)B_1^n$. Moreover, for $\eps < CRd/32\sqrt{n}$, $FX'$ is a $CRd\sqrt{d}/8\eps$-packing.

Now assume that $M=\{\mu_x\colon x\in\Rn\}$ is an $\eps$-differentially private
mechanism with error $CRd\sqrt{d}/32\eps$ and lead this to a
contradiction.
By the assumption on the error, Markov's inequality implies that for all $x\in X'$,
$\mu_x(B_x)\ge\tfrac12,$
where $B_x$ is a ball of radius $CRd\sqrt{d}/16\eps$ centered at $Fx$. By the packing property above, the balls $\{B_x: x\in X\}$ are disjoint.

Since $\|x'-v\|_1\le (d/2\eps)$, it follows from $\e$-differential privacy with
Fact~\ref{fact:trans} that
\[
\mu_v(B_x)
\ge\exp(-\eps (d/2\eps))\mu_{x}(B_x)\ge\tfrac12\exp(-d/2).
\]
Since the balls $B_x$ are pairwise disjoint,
\begin{equation}
1\ge \mu_0(\cup_{x\in X}B_x)
=\sum_{x\in X}\mu_0(B_x)
\ge \exp(d)\tfrac12\exp(-d/2)
> 1
\end{equation}
for $d\geq 2$. We have thus obtained a contradiction.
\end{proof}

Translating the lower bound from Theorem~\ref{thm:lower} to this setting, we get

\begin{theorem}
Let $\eps\in(0,(c\sqrt{(d/n)})\cdot \min\{\sqrt{d},\sqrt{\log (n/d)}\})$ for a universal constant $c$ and let $d \leq \log n$. Then there exists a linear map $F\in \{-1,1\}^{d\times n}$ such that every $\eps$-differentially private mechanism $M$ for computing $G(w)=Fx(w)$
must have
\begin{equation}
\err(M,G) \ge \Omega(d/\eps)\cdot \min\{\sqrt{d},\sqrt{\log (n/d)} \}.
\end{equation}
\end{theorem}

We remark that this lower bound holds for $N=\Omega(nd/\eps)$.

\section{Dilated Ball containment}
\begin{lemma}
\label{lem:ballcontainment}
Let $A$ be a convex body in $\Re^d$ such that $B_2^d \subseteq A + r B_2^d$ for some $r < 1$. Then a dilation $(1-r)B_2^d$ is contained in $A$.
\end{lemma}
\begin{proof}
Let $z \in \Re^d$ be a unit vector. Suppose that $z' = (1-r)z \not\in A$. Then
by the Separating Hyperplane theorem~(see, e.g., \cite{BoydV04}), there is a
hyperplane $H$ separating $z'$ from $A$. Thus there is a unit vector $w$ and a
scalar $b$ such that $\langle z',w\rangle-b=0$ and $\langle u,w\rangle-b \leq
0$ for all $u \in A$.  Let $v= z'+ rw$. Then by triangle inequality, $\|v\|
\leq 1$. Moreover,
\[
d(v,A) = \inf_{u\in A} \|u-v\|
\geq \inf_{u\in A}  \langle v-u,w\rangle
\geq b+r - \sup_{u\in A} \langle u,w\rangle
\geq r.
\]
This however contradicts the assumption that that $v \in B_2^d \subseteq A + rB_2^d$. Since $z$ was arbitrary, the lemma is proved.
\end{proof}

\end{document}